\documentclass[11pt, letterpaper]{amsart}
\usepackage{graphicx, hyperref}

\usepackage{amssymb,graphicx,float,bbm,amssymb,mathtools,epsfig,epstopdf}
\usepackage[a4paper]{geometry}
\usepackage{setspace, color}

\usepackage{enumerate}

\newtheorem{theorem}{Theorem}[section]
\newtheorem{lemma}[theorem]{Lemma}
\newtheorem{proposition}[theorem]{Proposition}
\newtheorem{corollary}[theorem]{Corollary}
\usepackage{amsmath}

\newtheorem{remark}[theorem]{Remark}
\theoremstyle{definition}
\newtheorem{definition}[theorem]{Definition}
\newtheorem{assumption}[theorem]{Assumption}

\theoremstyle{definition}

\usepackage{lipsum} 

\allowdisplaybreaks

\DeclareMathOperator*{\esssup}{ess\,sup}

\newcommand{\Rbb}{{\mathbb R}}
\newcommand{\Ebb}{{\mathbb E}}
\newcommand{\Fbb}{{\mathbb F}}
\newcommand{\Hc}{\mathcal{H}^2}
\newcommand{\Hci}{\mathcal{H}^{\infty}}

\newcommand{\Ac}{{\mathcal A}}
\newcommand{\Pbf}{{\mathbf P}}
\newcommand{\Fc}{{\mathcal F}}

\newcommand{\Wc}{{\mathcal W}}

\newcommand{\de}{{\delta}}
\newcommand{\De}{{\Delta}}
\newcommand{\bde}{{\bar \de}}

\newcommand{\Si}{{\Sigma}}
\newcommand{\si}{{\sigma}}
\newcommand{\la}{{\lambda}}
\newcommand{\La}{{\Lambda}}

\newcommand{\pnf}{{\widehat{\phi}}} 
\newcommand{\pnfp}{{\widetilde{\phi}}} 
\newcommand{\pnfn}{{\breve{\phi}}} 
\newcommand{\plm}{{\phi_{\La,m}}} 
\newcommand{\pdlm}{{\dot{\phi}_{\La,m}}} 
\newcommand{\pfp}{{\widetilde{\phi}_{\La,n}}} 
\newcommand{\pfdp}{{\dot{\widetilde{\phi}}_{\La,n}}} 
\newcommand{\pfnm}{{\breve{\phi}_{\La,m}}} 
\newcommand{\pfdnm}{{\dot{\breve{\phi}}_{\La,m}}}

\newcommand{\m}{{\mu}} 
\newcommand{\mnfp}{{\m_{n}}} 
\newcommand{\mnfn}{{\breve{\m}}} 
\newcommand{\mfn}{{\breve{\m}_{\La}}} 

\newcommand{\z}{\zeta} 
 
\newcommand{\II}{\mathcal{I}}
\newcommand{\DD}{\mathcal{D}}

\begin{document}
\title[Continuous-time Equilibrium in Markets with Price Impact \& Transaction Costs]{Continuous-time Equilibrium Returns in Markets with Price Impact and Transaction Costs}

\author{Michail Anthropelos}
\address{Department of Banking and Financial Management\\
University of Piraeus\\
Piraeus, Greece}
\email{anthropel@unipi.gr}
\thanks{We would like to thank Alexandra Chronopoulou, Constantinos Kardaras, Johannes Muhle-Karbe, Sergio Pulido, Scott Robertson and Gordan \v{Z}itkovi\'{c} for their valuable comments and suggestions.}

\author{Constantinos Stefanakis}
\address{Department of Banking and Financial Management\\
University of Piraeus\\
Piraeus, Greece}
\email{kstefanakis@unipi.gr}
\thanks{C. Stefanakis is supported in part by the Research Center of the University of Piraeus}

\begin{abstract}
We consider an It\^{o}-financial market at which the risky assets' returns are derived
endogenously through a market-clearing condition amongst heterogeneous risk-averse investors with quadratic preferences and random endowments. Investors act strategically by taking into account the impact that their orders have on the assets' drift. A frictionless market and an one with quadratic transaction costs are analysed and compared. In the former, we derive the unique Nash equilibrium at which investors' demand processes reveal different hedging needs than their true ones, resulting in a deviation of the Nash equilibrium from its competitive counterpart. Under price impact and transaction costs, we characterize the Nash equilibrium through the (unique) solution of a system of FBSDEs and derive its closed-form expression. We furthermore show that under common risk aversion and absence of noise traders, transaction costs do not change the equilibrium returns. On the contrary, when noise traders are present, the effect of transaction costs on equilibrium returns is amplified due to price impact. 
\ \\ 




\end{abstract}

\maketitle

\vspace{1.5cm}

\bigskip

\section{Introduction}

It has been an undeniable fact that in several contemporaneous financial markets there are large institutional investors whose trading strategies impact the assets' prices (see among others the empirical studies in \cite{KoiYog19} and \cite{HuHofLanTim21}). Motivated by the related statistical evidence, a large body of theoretical research on price-impact modelling has been recently developed both in static and dynamic settings (for extended literature reviews on price-impact models we refer to \cite{RosYoo22} and \cite{price_impact}).   

The majority of the existing models in this area, however, do not take into account a prevalent force in the markets: the exogenous transaction costs. Especially when a market is thin, i.e.~when the average trading volume is low, transaction costs bedevil the investors in several forms, such as broker and trading platform's  fees (\cite{FraIsrMosk18}) and searching-time for counterpart (\cite{DuffGarPed05}). Such transaction costs can be seen as an additional type of market's frictions on top of its non-competitiveness.  

In this paper, we study a continuous-time market populated by strategically behaved investors who exploit their price impact, while being charged with transaction costs on their trading. More precisely, we assume that risky assets follow It\^{o} dynamics, with drifts that are endogenously derived at the equilibrium. The presence of noise traders is allowable in the model, while the main market's participants are mean-variance investors endowed with heterogeneous random income processes. Driven by their hedging needs, investors trade the risky assets and the market-clearing condition determines the equilibrium drift of assets' return. In contrast to the related literature (see among others \cite{Bas00}), we assume (as in \cite{RosWer15}) that all investors are strategic, in the sense that each one takes into account her impact to market equilibrium when submitting an order. Unlike \cite{RosWer15} however, we impose (as in \cite{C}) quadratic transaction costs on investors' trading, which we argue it is a more appropriative market setting for thin financial markets. To the best of our knowledge, this is the first attempt to model a continuous-time market equilibrium with the simultaneous presence of transaction costs and endogenous price impact.

\subsection*{Main contributions}

The paper analyses and compares two main market settings: one without transaction costs (also called frictionless) and another with the presence of transaction costs (also called frictional market). We use the frictionless market as a benchmark to highlight the effect of transaction costs on market equilibrium, but also to introduce in a more clear manner the way we model investors' price impact. For this, we write the market's equilibrium return as a function of each investor's demand given that the rest of the investors act optimally (similarly to the static models of \cite{A} and \cite{B}). This creates a best-response function which in turn leads to a continuous-time Nash equilibrium for market returns. 

We derive these equilibrium returns in closed-form and verify that investors' price impact  does always change the market's drift when noise traders are present (even under common risk aversion). In fact, at the equilibrium trading, price impact makes the investors reveal different hedging needs than their true ones. In particular, this leads to lower non-competitive equilibrium returns (and hence higher equilibrium prices) when noise traders buy the risky assets (from equally risk tolerant investors). Consistent to the related literature (\cite{MalRos17} and \cite{AnthKard24}), we also show that Nash equilibrium produces better utility gains for investors with sufficiently low risk aversion, when compared with the associate competitive equilibrium. 

We then apply this price-impact concept in a market with quadratic transaction costs. For such extension, we have to deal with a number of challenges. It is well known that equilibrium models with transaction costs are intractable, because trading costs complicate the investors' individual optimization objectives. In our case, the problem is compounded since there is no general explicit expression of the competitive equilibrium returns in the market, which means that there is no direct way to identify investors' price impact. For this, we restrict the class of admissible strategies and model's parameters and assume common risk tolerance, which in turn leads to the (frictional) best-response strategy satisfying an explicitly solvable coupled but linear FBSDE. Building on this, we characterize the investors' Nash equilibrium demands through an associate system of linear FBSDEs, which admits an explicit solution. We then derive a closed-form expression of the (frictional) Nash equilibrium returns.

The closed-form formulas allow us to compare the market's return under different types of equilibria in two directions: with and without transaction costs and with and without price impact. The impact on investors' frictional objectives comes both in their linear and quadratic parts. The linear part exploits solely the direction of noise traders demand, while the quadratic part indicates the negative effect of transaction costs on investor's strategic behavior. On the other hand, we show that investors' price impact increases the part of the equilibrium return that stems from transaction
costs (when more than two investors participate in the market). In the special case of only two
strategic investors, the effect of transaction costs on equilibrium returns is equal between competitive and non-competitive market's structure. We also highlight the importance of homogeneous risk preference on equilibrium comparison. In particular, when investors have common risk aversion and there is no noise traders in the market, we show that frictional and frictionless Nash equilibrium coincide (a result that holds under competitive market structure too).


\subsection*{Connection with the related literature}

Our paper is linked with two strands of the literature: equilibrium models with price impact and the optimal investment and equilibrium pricing under transaction costs. 

There different categorisations of price-impact models (with or without time-dynamic setting) regarding the source of price impact and the (a)symmetry of participating investors. On one hand, there is a vast literature that assumes exogenous price impact, e.g.~\cite{AlmCh00}, \cite{AlmThu05}, \cite{HubSta04} and \cite{SchZha19}. Therein, specific assumptions are used to justify the nature of the price impact that could be permanent or temporary.  In contrast to this approach, we follow the literature that starts from the seminal work of \cite{Kyl85} and considers an endogenously derived price impact. In fact, similarly to the static models of \cite{A}, \cite{Viv11} and \cite{MalRos17} and the time-dynamic ones in \cite{Vay99} and \cite{RosWer15}, all investors, except the noise ones, have symmetric bargaining power and internalize their ability to affect the equilibrium returns. 

There are two main advantages of imposing endogenous price impact: the first is that there is no need to impose an exogenous parametrization of impact and second that we can analyse the equilibrium price impact in terms of investors' characteristics such as their risk exposures, risk preferences and/or even their population. 

Our work contributes also on equilibrium market models with transaction costs incorporating a non-competitive market structure. Under that perspective, we extend the equilibrium models of \cite{C} and \cite{D} and show that price-impact not only can be taken into account, but also the equilibrium outcomes can be analytically and qualitatively compared. Equilibrium models with dynamic trading and transaction costs, however without assuming price-impact, have been developed both in discrete-time setting (e.g.~\cite{BusDum19}) and under continuous-time (e.g.~\cite{Vay98}). Equilibrium continuous-time models with transaction costs are also developed in \cite{VayVil99} and \cite{West18}, but under simplified settings with deterministic asset prices.

On the other hand, Nash equilibria under price impact and transaction costs in an dynamic setting have studied in \cite{SchZha19} and \cite{LuoSch19} and more recently extended in \cite{CorLil24}, under the notion of price impact game. In contrast to our setting, these works consider exogenous price impact, while the trading occurs in discrete times.

\subsection*{Structure of the paper}
This paper is organized as follows: Section \ref{sec:setup} presents the model's ingredients and the simple case of the competitive equilibrium with no transaction costs. In Section \ref{sec:Price_Impact_NoFri}, we introduce the price impact modelling in a frictionless market, derive the closed-form expression of the Nash equilibrium returns and comment on its intuition. In Section \ref{sec:Frictions}, we generalize the model by including quadratic transaction costs and solve the system of FBSDEs that characterizes the corresponding Nash equilibrium returns. Finally, in Appendix \ref{sec:appendixA} we provide the proofs of Section \ref{sec:Price_Impact_NoFri}, while the proofs of Section \ref{sec:Frictions} are given in Appendix \ref{sec:appendixB}. 



\bigskip
\section{The Model Set-up}\label{sec:setup} \normalfont

We fix a probability space $\big(\Omega,\Fc,\Pbf\big)$ equipped with an augmented filtration $\Fbb:=(\Fc(t))_{t\in[0,T]}$ generated by a standard $d$-dimensional Brownian motion $W$, and a finite time horizon $T>0$. We also fix a constant $r\geq 0$, which will stand for the time discount rate. 

\subsection{The market}

We consider a market with $d+1$ assets and their cumulative return processes $R=(R^{i}:0\leq i \leq d+1)$. The first one is exogenous and holds no risk, while the remaining assets are risky with returns modeled by the following $\Rbb^{d}$-valued It\^{o} process:
\begin{equation} \label{eq:returns}
\begin{bmatrix}dR^{1}(t) \\ \vdots \\ dR^{d}(t)\end{bmatrix}=\begin{bmatrix}\nu^{1}(t)\\ \vdots \\ \nu^{d}(t)\end{bmatrix}dt+\begin{bmatrix}\si_{1,1}&\cdots&\si_{1,d}\\ \vdots& &\vdots \\\si_{d,1}&\cdots&\si_{d,d}\end{bmatrix} \begin{bmatrix}dW^{1}(t) \\ \vdots \\ dW^{d}(t)\end{bmatrix}, \ \ \ R^{i}(0)=0, \  i\in\DD,
\end{equation}
where $\DD:=\{1,2,...,d\}$. Here, the $\Rbb^{d}$-valued (local) return process $\nu=(\nu^{i}:i\in\DD)$ is going to be endogenously determined at the market equilibrium, while the constant $\Rbb^{d\times d}$-valued infinitesimal covariance matrix $\Si:=\si\si^{'}$ is given exogenously and is assumed to be (symmetric) positive definite. 
Through \eqref{eq:returns} we also characterize the dynamics of the risky assets' price processes $S=(S^{i}:i\in\DD)$ as:
\begin{equation}\label{eq:stocks}
dS^{i}(t)=S^{i}(t)dR^{i}(t), \quad i\in\DD,
\end{equation}
together with the specified initial values $(S^{i}(0):i\in\DD)$ which are strictly positive. 
Regarding the riskless asset we require it pays and charges no interest, i.e.~we set $R^{0}= 0$ and $S^{0}= 1$. This means that we can consider the wealth and endowment processes at discounted terms \footnote{This is done for the sake of clarity since we would arrive at the same results by taking $R^0$ to be an arbitrary adapted, continuous process of finite variation with $R^0(0)=0$, $S^0$ satisfying: $dS^{0}(t)=S^{0}(t)dR^{0}(t)$ for $S^{0}(0)=1$ and considering $R^{i}-R^{0}, \ S^{i}/S^0, \ i\in\DD$, as our respective return and price processes, which once more satisfy \eqref{eq:stocks}. We can summarize the aforementioned as normalizing the riskless asset price process to one, and ``expressing everything else in its units", through discounting.}. 

\subsection{The investors}
We consider $N$ investors with quadratic preferences over their wealth at time terminal time $T$ and risk tolerance coefficients denoted by $\de_{m}>0$ for each $m\in \II:=\{1,2,...,N\} $. They are about to continuously trade in the market in order to hedge their exogenously given risks that stem from the fluctuations of their cumulative endowment processes:
\begin{equation}\label{eq:endowment}
dY_{m}(t)=(\z_{m}(t))^{'}\si dW(t), \ \ \ m\in\II. 
\end{equation}
In the above formulation of endowment, $W$ stands for the $d$-dimensional risk factor and $\z_{m}=(\z_{m}^{i}:i\in\DD)$ is an adapted process that describes the exposure of investor $m$ to the risk factor. The stochasticity of $\z_{m}$ for each $m\in\II$ implies that investors' exposures to market's risk factor change through time and hence their hedging needs change accordingly. In principle, this implies that even if the equilibrium transaction among investors is Pareto optimal, endowments' changes give rise to additional room for mutually beneficial trading. 

\begin{remark}
Considering only the diffusion part of the investors' random endowment in \eqref{eq:endowment} comes without loss of generality. Indeed, as stated in \cite{D}, since the solution to the investors' objective functional (see \eqref{eq:frless_objective} below) is invariant to any additional finite variation part or to endowment's shocks that are orthogonal to the market. In other words, under suitable integrability assumptions, \eqref{eq:endowment} could be generalized to:
\begin{equation*}
dY_{m}(t)=dA_{m}(t)+(\z_{m}(t))^{'}\si dW(t)+dM_{m}^{\perp}(t), \qquad m\in\II, 
\end{equation*}
where the $\Rbb$-valued, finite variation process $A_{m}$ models the cumulative cashflows not spanned by the assets and the $\Rbb$-valued orthogonal martingale $M_{m}^{\perp}$ stands for any unhedgeable shocks in the investors' endowment. 
\end{remark}

The amount of money that each investor invests on the risky assets is given by the $\Rbb^{d}$-valued processes $\phi_{m}=(\phi_{m}^{i}:i\in\DD)$, for each $m\in\II$, which shall be called  portfolios or demands.  Combining trading in the market with the endowment process and assuming no additional withdrawn or deposit of money until time $T$, we get that the dynamics of the each investor's total wealth is of the following form:
\begin{equation}\label{eq:wealth process}
dV_{m}(t)=(\phi_{m}(t))^{'}dR(t)+dY_{m}(t), \ \ \ 1\leq m \leq N.
\end{equation}

Although, it is not a necessity for the model, we consider the potential participation of noise liquidity providers, hereafter called noise traders, whose demand is exogenously given and denoted by the $\Rbb^{d}$-valued process $\psi$.

We will use the following class of processes for any $1\leq p <\infty$:
\begin{equation*}
\mathcal{H}^{p}:=\bigg\{X:[0,T]\times\Omega\rightarrow\mathbb{R}^{d}: \text{$X$ optional s.t.} \ \|X\|_{\mathcal{H}^{p}}:=\mathbb{E}\bigg[\int_0^{T}\|X(t)\|_{2}^{p}dt\bigg]^{1/p}<\infty\bigg\},
\end{equation*}
and impose the following standing assumption:
\begin{assumption}\label{ass:standing}
$\psi,\phi_{m},\nu$ and $\z_{m}\in \Hc,$ for all $ m \in\II.$
\end{assumption}

\begin{remark}
Note that we could also represent the investors' holdings in units of assets as: $\theta_{m}^{i}:=\phi_{m}^{i}/S^{i}$ for each $i\in\DD$ $ m \in\II$. That is, by the associativity of the semimartingale integral 
we have that \eqref{eq:wealth process} can also be written as:
\[
dV_{m}(t)=(\theta_{m}(t))^{'}dS(t)+dY_{m}(t), \quad  m \in\II.\]
\end{remark}

\subsection{Equilibrium returns in a competitive market with no transaction costs}

We start our equilibrium study with the simplest case of a competitive market without transaction costs. Although this equilibrium is not our main focus, we use it as a benchmark for equilibrium comparisons that helps the related intuition discussion. Herein, the objective of investor $m$ is to choose a portfolio $\phi$ through the maximization of the following quadratic functional $\mathcal{F}_{m}:\Hc\rightarrow\Rbb$ defined as:
\begin{equation*}
\mathcal{F}_{m}(\phi):=\Ebb\left[\int_0^{T}e^{-r t}\left((\phi(t))^{'}dR(t)+dY_{m}(t)-\frac{1}{2\de_{m}}d\left[ \int_0^{.}(\phi)^{'}dR(s)+Y_{m} \right](t)\right)\right], \quad m\in\II,
\end{equation*}
where the wealth dynamics are given by \eqref{eq:wealth process}. Hence, the optimization problem is written as:
\begin{equation}\label{eq:frless_objective}
\sup\limits_{\phi\in\Hc}\Ebb\bigg[\int_0^{T}e^{-r t}\Big((\phi(t))^{'}\nu(t)-\frac{1}{2\de_{m}}(\phi(t)+\zeta_{m}(t))^{'}\Si (\phi(t)+\zeta_{m}(t))\Big)dt\bigg], \quad m\in\II.
\end{equation}

\begin{remark}\label{rem:r}
The presence of the discount factor $e^{-rt}$ in optional, since we could simply set $r=0$. A strictly positive $r$ implies the presence of a penalizing term on investors' objective functional when the wealth is utilized later in time.  
\end{remark}

By the strict concavity of $\Fc_{m}(\phi)$ the unique solution of  \eqref{eq:frless_objective}, denoted by $\pnf_{m}$, is derived via a simple calculus of variations argument, 
and has the following representation:
\begin{equation}\label{eq:frless_comp_strategies}
\Hc\ni\pnf_{m}:=\de_{m}\Si^{-1}\nu-\z_{m}, \quad\forall m\in\II.
\end{equation}
The above investment strategy shall be called \textit{the competitive frictionless optimal portfolio} of investor $m$. Note that the first term of \eqref{eq:frless_comp_strategies} is the classic quadratic optimal portfolio (also called Merton's portfolio), while the second reflects the investor's hedging needs. 

Under competitive market structure with no transaction costs, the local return process $\nu$ of \eqref{eq:returns} is determined simply by the market-clearing condition, where a zero-net supply is imposed:  
\begin{equation}\label{eq:equilibrium_condition}
\pnf_1(t)+\cdots+\pnf_{N}(t)+\psi(t)=0, \ \ \ \forall \ t\in[0,T].
\end{equation}
Without transaction costs and when all investors are price takers, we readily get from \eqref{eq:frless_comp_strategies} and \eqref{eq:equilibrium_condition}, the explicit form for the equilibrium return process, henceforth called \textit{frictionless competitive equilibrium returns}:
\begin{equation}\label{eq:frless_comp_equilibrium}
\m:=\frac{\Si(\zeta-\psi)}{\de},
\end{equation}
where the aggregate investors' risk exposure and risk tolerance are defined as:
\begin{equation}\label{eq:aggregate}
\z:=\sum_{m=1}^{N}\z_{m}\quad \text{ and }\quad\de:=\sum_{m=1}^{N}\de_{m}. 
\end{equation}
Due to Assumption \ref{ass:standing}, we directly verify that $\m\in\Hc$. Equilibrium \eqref{eq:frless_comp_equilibrium} is indicative for the discussion that follows. Say $d=1$, for simplicity. Assuming that the aggregate investors' exposure to market shocks is higher than the corresponding noise demand, the  equilibrium returns increase when aggregate risk tolerance decreases and the asset's volatility increases. This is intuitive: lower risk tolerance and higher risk means that investors decrease their demand for the asset and hence equilibrium returns increase. Note that higher equilibrium returns essentially imply that investors \textit{require} higher expected returns in order to invest on the asset, a situation that is consistent to lower demand. On the other hand, higher exposure to market shocks increases investors' demand and hence decreases the required expected returns at the equilibrium.   

\smallskip

\section{\large Price Impact in a Market with no Transaction Costs}\label{sec:Price_Impact_NoFri} \normalfont

We are now ready to introduce the way we model investors' \textit{price impact}. For this, we first take the position of investor $n$ and assume (for now) that she acts strategically, while the rest of the investors are price takers. This creates a best-response function that eventually leads to the Nash equilibrium when all investors apply the same strategy. 

Consider again the market clearing condition \eqref{eq:equilibrium_condition} and assume that all the investors, besides $n$, act optimally through \eqref{eq:frless_comp_strategies}. In the spirit of \cite{A}, we observe that market's returns can be viewed as a function of investor $n$'s demand. Indeed, in this context the clearing condition is written as:
\begin{equation}\label{eq:best_response}
\phi_{n}(t)+\sum_{\substack{m=1\\m\neq n}}^{N}\Big(\de_{m}\Si^{-1}\nu(t)-\z_{m}(t)\Big)+\psi(t)=0, \ \ \ \forall \ t\in[0,T], \ \ \ \forall \ \phi_{n}\in\Hc.
\end{equation}
Solving for $\nu$, we get the \textit{price impact process} of investor $n$, $\nu_n(\phi)\in\Hc$:
\begin{equation}\label{eq:impact}
\nu_n(t;\phi)=\frac{\Si(\z_{-n}(t)-\psi(t)-\phi(t))}{\de_{-n}}, \ \ \ \forall \ t\in[0,T], \ \ \ \forall \ \phi\in\Hc,
\end{equation}
where $\z_{-n}:=\sum_{\substack{m=1\\m\neq n}}^{N}\z_{m}$ and $\de_{-n}:=\sum_{\substack{m=1\\m\neq n}}^{N}\delta_{n}$ denote respectively the aggregate exposure to market shocks and the aggregate risk tolerance of the rest of the investors. Under a slight abuse of notation we distinguish between the price-impact process of an investor $n$, $\nu_n(\cdot)$ and the local market's return process $\nu$, which is to be determined at the equilibrium. Functional $\nu_n(\phi)$ indicates that when investor $n$ submits positive demand for the asset, the equilibrium return decreases at a rate of $\Sigma/\delta_{-n}$. In turn, the way investor $n$ chooses the optimal strategy is by maximizing \eqref{eq:frless_objective} but after taking into account her impact on equilibrium returns through \eqref{eq:impact}. Hence, the adjusted objective functional for investor $n$, $\widetilde{\Fc}_{n}:\Hc\rightarrow\Rbb$, is given as:
\begin{equation}\label{eq:best_response_objective}
\widetilde{\Fc}_{n}(\phi):=\Ebb\bigg[\int_0^{T}e^{-r t}\Big((\phi(t))^{'}\nu_{n}(t;\phi)-\frac{1}{2\de_{n}}(\phi(t)+\z_{n}(t))^{'}\Si (\phi(t)+\z_{n}(t))\Big)dt\bigg].
\end{equation}
Recalling \eqref{eq:impact} and rearranging the terms, we get the following optimization problem:
\begin{equation}\label{eq:best_response_problem}
\sup\limits_{\phi\in\Hc}\Ebb\Bigg[\int_0^{T}e^{-r t}\bigg((\phi(t))^{'}\Si\Big(\frac{\z_{-n}(t)-\psi(t)}{\delta_{-n}}-\frac{\z_{n}(t)}{\de_{n}}\Big)-\frac{k_{n}}{2}(\phi(t))^{'}\Sigma\phi(t)-\frac{1}{2\de_{n}}(\z_{n}(t))^{'}\Si\z_{n}(t)\bigg)dt\Bigg], 
\end{equation}
where 
$$k_{n}:=\frac{1}{\de_{-n}}+\frac{1}{2\de_{n}}.$$ 
To simplify the formulas that follow, we introduce the notation for each investor's relative risk tolerance:
\begin{equation*}
\begin{aligned}
&\la_{m}:=\frac{\de_{m}}{\de} \ \ \ \text{and} \ \ \ \la_{-m}:=1-\la_{m}, \  \ \ \ m\in\II.
\end{aligned}
\end{equation*}
The following proposition deals with the solution of problem \eqref{eq:best_response_problem}.
\begin{proposition}\label{prop:best_response}
Under Assumption \ref{ass:standing}, the optimization problem \eqref{eq:best_response_problem} has the following unique solution:
\begin{equation}\label{eq:best_response_solution}
\Hc\ni\pnfp_{n}:=\frac{\la_{n}(\z_{-n}-\psi)-\la_{-n}\z_{n}}{\la_{n}+1}.
\end{equation}
\end{proposition}
Henceforth, we also refer to \eqref{eq:best_response_solution} as the \textit{best-response strategy} of investor $n$ in the frictionless market. 
\begin{remark}
Comparing the best-response strategy of investor $n$ with the respective competitive optimal strategy \eqref{eq:frless_comp_strategies} at the competitive equilibrium we directly get:
\begin{equation*}
\pnfp_{n}(t)=\frac{1}{\la_{n}+1}\pnf_{n}(t), \ \ \ \forall \ t\in[0,T].
\end{equation*}
Hence, investor $n$ has motive to decrease the volume of her demand (positive or negative) when compared to her price-taking demand. Intuitively, lower demand implies a lower equilibrium price and hence higher required return for the assets that the investor is about to buy. 

On the other hand, it holds that
\begin{equation}\label{eq:best_response_return}
\Hc\ni\mnfp:=\frac{\la_{n}}{\la_{n}+1}\m_{-n}+\frac{1}{\la_{n}+1}\m, 
\end{equation}
where $\m_n$ is the equilibrium returns when only investor $n$ acts strategically, $\m_{-n}:=\Si(\z_{-n}-\psi)/\de_{-n}$ is the competitive equilibrium without the presence of the strategic investor $n$ and $\m$ is the competitive equilibrium when all investors are present (given in \eqref{eq:frless_comp_equilibrium}). This implies that the equilibrium returns are the average return of the competitive equilibrium $\m$ for all the market and the competitive equilibrium without investor $n$, meaning that the impact of investor $n$'s strategy drives the equilibrium return between $\m_{-n}$ and $\m$.
\end{remark}

\subsection{Investors' revealed risk exposure}\label{sub:revealed_exposure}

By applying strategy \eqref{eq:best_response_solution}, investor $n$ drives the equilibrium returns to \eqref{eq:best_response_return}, i.e.~to a different equilibrium. As mentioned in the introductory section, the main reason for investors' trading in this market is to hedge their risk exposure (the other reason is to offset noise traders' demand). When investor $n$ applies strategy \eqref{eq:best_response_solution}, she essentially reveals different hedging needs than her true ones. 
In order to see that, recall the price-taking demand \eqref{eq:frless_comp_strategies}, where the risk exposure is the intercept point of the affine demand function. Under price impact, investor $n$ submits a different demand function and hence the implied risk exposure is shifted. This creates the very useful and intuitive notion of the \textit{revealed risk exposure} that is formally defined below. 
\begin{definition}[Revealed risk exposure process]\label{def:revealed_exposure}
For any given demand process $\phi\in\Hc$ of investor $m$ and local returns $\nu\in\Hc$ pinned down at the equilibrium, the \textit{revealed exposure} of investor $m$ is defined as: 
\begin{equation}\label{eq:revealed_exposure}
\z_m(\phi):=\de_{m}\Si^{-1}\nu-\phi, \ \ \ \forall  m \in\II.
\end{equation}
\end{definition}


With a slight abuse of notation, we distinguish between the revealed risk exposure process $\z_{m}(\cdot)$ and the true exposure process, $\z_{m}$ of each investor $m$. Departing from the above definition, we derive the following result: 
\begin{proposition}\label{prop:best_response_revealed}
Impose Assumption \ref{ass:standing}. The revealed exposure from the equilibrium returns $\mnfp$ and the demand $\pnfp_{n}$ of investor $n$ is given as: 
\begin{equation}\label{eq:best_response_revealed}
\z_{n}(t;\pnfp_n)= \frac{\la_{n}^2}{1-\la_{n}^2}(\z_{-n}(t)-\psi(t))+\frac{1}{1+\la_{n}}\z_{n}(t),  \ \ \ \forall \ t\in[0,T].
\end{equation}
\end{proposition}
\begin{remark}
The coefficient of the second term in \eqref{eq:best_response_revealed} is positive and less than one. Whence, investor $n$'s revealed risk exposure is a fraction of her true one and a part of the net aggregate exposure $\zeta_{-n}-\psi$. The latter can be called residual exposure, since it stands for the aggregate exposure to market of the rest of the investors (after the offsetting of the noise traders' demand). In other words, strategic investor $n$ on the one hand lessens the true exposure of her endowment and on the other supplements it by a fraction $\propto (\zeta_{-n}-\psi)$. We could think of the above strategy as investor $n$ declaring only a part of her true exposure to the rest of the investors, while exploiting her impact on the equilibrium by submitting exposure similar to her counterparties.
\end{remark}

A reasonable question that arises is when does the strategic demand of an investor coincides with the competitive one. As stated in the following corollary, this coincidence never happens, as long as investor's demand at competitive equilibrium is non-zero.
\begin{corollary}\label{cor:fr_strategies_equality}
Under Assumption \ref{ass:standing}, it holds that: 
\begin{equation*}
\pnf_{n}(t)=\pnfp_{n}(t),\quad \forall \ t\in[0,T],\qquad\text{ if and only if }\qquad \pnf_{n}(t)=0,  \quad \forall \ t\in[0,T].
\end{equation*}
\end{corollary}
Note that it is rather standard to have $\pnf_{n}(t)\neq 0$. Indeed, from \eqref{eq:frless_comp_strategies} and \eqref{eq:frless_comp_equilibrium}, we get that condition $\pnf_{n}=0$ holds if and only if $\lambda_n(\zeta_{-n}-\psi)=(1-\lambda_n)\zeta_n$. Even if the investors' endowment have a specific linear dependence, the presence of noise traders makes the condition extremely rare.

\subsection{The Nash equilibrium returns in a market with no transaction costs}\label{sub:2}

We now consider a market where all the investors act strategically by applying the strategy that is determined by problem \eqref{eq:best_response_objective}. In the view of \eqref{eq:revealed_exposure}, the Nash equilibrium is induced as the fixed point of \eqref{eq:best_response_revealed}. We formally define the the Nash equilibrium returns in this context below:
\begin{definition}\label{def:nash_frless}
The processes $(\pnfn_{m}\in\Hc: m\in\II)$ are called investors' \textit{Nash optimal demands} if for each $m\in\II$, $\pnfn_{m}$ corresponds through \eqref{eq:revealed_exposure} to the revealed risk exposures of  \eqref{eq:best_response_revealed}.
In turn, the return process induced by the market-clearing condition when all the investors act according to their Nash optimal demands, is called \textit{Nash equilibrium returns}. 
\end{definition}


In the view of \eqref{eq:best_response_solution}, \eqref{eq:revealed_exposure} and the investors' symmetry, their optimal demands $(\pnfn_{m}\in\Hc: m\in\II)$ solve the following linear system:
\begin{equation}\label{eq:Nash_system_frless}
\pnfn_{m}=\frac{\la_{m}(\z_{-m}(\pnfn)-\psi)-\la_{-m}\z_{m}}{\la_{m}+1}, \quad m\in\II.
\end{equation}
The fact that solutions to the above belong in $\Hc$ is guaranteed by Assumption \ref{ass:standing} and system's linearity. In other words, Nash equilibrium demands have the form of the best-response strategy \eqref{eq:best_response_solution}, but with the rest of the investors being strategic too. Armed with the above, we give the following characterization:
\begin{theorem}\label{the:Nash_frless}
Under Assumption \ref{ass:standing}, the unique Nash equilibrium returns in a market without transaction costs take the following form:
\begin{equation}\label{eq:Nash_frless}
\Hc\ni\mnfn:=\frac{\Si}{\de}\frac{(\z-\psi)-\sum_{m=1}^{N}\la_{m}\z_{m}}{1-\sum_{m=1}^{N}\la_{m}^2}.
\end{equation}
\end{theorem}

As Corollary \ref{cor:fr_strategies_equality} implies, Nash equilibrium is generally different than the competitive one. Indeed, using \eqref{eq:frless_comp_equilibrium} and \eqref{eq:Nash_frless} we readily get that:
\begin{equation}\label{eq:Nash_competitive_frless}
\mnfn(t)-\m(t)=\frac{\Si}{\de}\frac{1}{1-\sum_{m\in\II}\la_{m}^2}\sum_{m\in\II}\la_{m}\pnf_{m}(t), \quad \forall \ t\in[0,T],
\end{equation}
where we recall that $\pnf_{n}$ is the investor $n$'s optimal demand at the competitive equilibrium, given in \eqref{eq:frless_comp_strategies}. The difference of equilibrium returns with and without price impact given in \eqref{eq:Nash_competitive_frless} also called \textit{liquidity premium} created by the presence of price impact. Similarly to the condition $\pnf_n=0$, the weighted sum $\sum_{m\in\II}\la_{m}\pnf_{m}$ rarely equals to zero. For example, when investors have the same risk tolerance, i.e.~$\de_m=\bde$, $\forall m\in\II$,  
the liquidity premium equals to $-(\Si/\bde N(N-1))\times \psi$, which vanishes if and only if there is no noise demand in the market. In other words, there is no liquidity premium due to price impact only when investors' risk tolerances are equal and there is no noise trader. 

On the other hand, even if $\psi=0$, then $\mnfn=\m$ holds if and only if $\zeta=\sum_{m\in\II}\lambda_m\zeta_m/\sum_{m\in\II}\lambda_m^2$, which is a quite rare condition. 

The effect of price impact on the equilibrium returns has also a monotonicity property. In particular, we get from \eqref{eq:Nash_competitive_frless} that when $\de_m=\bde$, $\forall m\in\II$: $(\mnfn-\m)'\psi\leq 0$, which means that price impact creates liquidity premium on the assets' required return in the opposite direction to noise traders demand. This is apparent when $d=1$: A positive (resp.~negative) noise demand decreases (resp.~increases) the equilibrium required asset return, or equivalent increases (resp.~decreases) its price. This implies that the noise traders will buy (resp.~sell) the asset at a higher (resp.~lower) price (an effect that benefits the investors in the aggregated level).

Another issue that stems from equilibria comparison is whether the utility gain (or utility surplus) that is created by the trading is higher or lower for each investor due to price impact. In general, competitive equilibrium is consistent with Pareto optimal allocation and hence another equilibrium allocation decreases the aggregate gain of utility. However, in individual level, Nash equilibrium may be beneficial for some investor(s). This feature appears in several recent non-competitive settings such as \cite{AnthKard24} and \cite{MalRos17}. 

To this end, we define the utility surplus for each asset return $\nu\in\Hc$ and each admissible strategy $\phi\in\Hc$ of the investor $m$ as:
\begin{equation}\label{utilsurp}
\mathbb{US}_{m}^{\nu}(\phi):=\mathcal{F}_{m}(\phi)-\mathcal{F}_{m}(0), \ \ \ \forall\ \  m\in\II.
\end{equation}
In the view of the related literature we get the following limiting result:
\begin{proposition}\label{prop:surplus_comparison}
Impose Assumption \ref{ass:standing} and fix $(\de_2,...,\de_N)\in(0,\infty)^{N-1}$. Then, we get:
\begin{equation*}
\lim_{\de_1\rightarrow\infty}\mathbb{US}_{1}^{\m}(\pnf_1) = 0
\end{equation*}
and on the other hand:
\begin{equation*}
\lim_{\de_1\rightarrow\infty}\mathbb{US}_{1}^{\mnfn}(\pnfn_1)=\Ebb\bigg[\int_0^{T}e^{-rt}\frac{1}{4\de_{-1}}(\z_{-1}(t)-\psi(t))'\Si(\z_{-1}(t)-\psi(t))dt\bigg]\geq 0.
\end{equation*}
\end{proposition}

Note that $\delta_1=\infty$ refers to risk-neutral investor (vanishing risk aversion). Due to positive definiteness of $\Si$, $\lim_{\de_1\rightarrow\infty}\mathbb{US}_{1}^{\mnfn}(\pnfn)$ is normally strictly positive and becomes zero if and only if the noise traders' demand exactly matches (in a $dt\otimes d\Pbf$ sense) the hedging demands of the rest of the investors. Therefore, the risk-neutral investor benefits from the price impact as long as noise traders demand is different than the rest of the investors' risk exposure. In fact, her utility gain  essentially stems from the aggregate difference $\z_{-1}(t)-\psi(t)$.

\smallskip

\section{\large Price Impact in a Market with Transaction Costs}\label{sec:Frictions} \normalfont

In this section we aim to generalize the concept of the Nash equilibrium under the presence of exogenous market transaction costs. As in the frictionless case, we build our price-impact equilibrium upon the corresponding competitive one. For this, we begin our analysis with a brief introduction on the competitive equilibrium when each investor faces transaction costs upon trading. 

\subsection{Competitive market with transaction costs}

 As mentioned in the introductory section, we consider market's transaction costs and more precisely in the form of quadratic costs on the rate of trading. Hence, the frictional (i.e.~with transaction costs) analogue of the competitive functional, denoted by $\Fc_{\La,m}$, is defined as:
\begin{equation}\label{eq:fr_objective}
\Fc_{\La,m}(\phi):=\Ebb\bigg[\int_0^{T}e^{-r t}\Big((\phi(t))^{'}\nu(t)-\frac{1}{2\de_{m}}(\phi(t)+\zeta_{m}(t))^{'}\Si (\phi(t)+\zeta_{m}(t))-(\dot{\phi}(t))^{'}\La\dot{\phi}(t)\Big)dt\bigg], 
\end{equation}
for each $m\in\II$, where $\dot{\phi}_{m}(t):=d\phi_{m}(t)/dt$ stands for investor's trading rate  and $\La$ is an $\Rbb^{d\times d}$-valued constant diagonal and positive definite matrix. In words, investors face the same costs that are proportional to the square of their trading rate. Furthermore, since the frictional competitive functional involves the investor's trading rate, a natural assumption to avoid infinite transaction costs is to impose $\dot{\phi}_{m}\in\Hc$, for each $m\in\II$. 
%

Focusing on demand functions of the form $\phi_{m}(t)=\int_0^{t}\dot{\phi}_{m}(s)ds, \ m\in\II$, we can state our maximization objective in optimal control terms. More precisely, for each $m\in\II$, $\phi_{m}$ can be viewed as a state variable which is gradually evolved by adjusting to the control $\dot{\phi}_{m}$. Hence, the objective of investor $m$ for a competitive market with transaction costs could be written as
\begin{equation}\label{eq:fr_problem}
\sup\limits_{\dot{\phi}_{m}\in\Hc}\Fc_{\La,m}(\dot{\phi}).
\end{equation}
From \cite[Lemma~4.1]{C}, we get that \eqref{eq:fr_problem} has a unique solution which is characterized by the following coupled and linear Forward-Backward SDE (FBSDE) for each $m\in\II$:
\begin{equation}\label{eq:fr_solution_compe}
\begin{aligned}
&d\phi_{m}(t)=\dot{\phi}_{m}(t)dt, \ \ \ \phi_{m}(0)=0, \\
&d\dot{\phi}_{m}(t)=dM_{m}(t)+\frac{\La^{-1}\Si}{2\de_{m}}(\phi_{m}(t)-\pnf_{m}(t))dt+r\dot{\phi}_{m}(t)dt, \ \ \ \dot{\phi}_{m}(T)=0,
\end{aligned}
\end{equation}
where $M_{m}\equiv(M_{i,m}: i\in\DD)$ is a $\Rbb^{d}$-valued, square integrable (continuous) $\Fbb$-martingale determined as part of the solution, and $\pnf_{n}$ is the frictionless competitive optimal given by  \eqref{eq:frless_comp_strategies}.

\subsection{The price impact of a single investor in a market with transaction costs}

We return to the concept of price impact imposing now transaction costs to the market. As in the frictionless equilibrium, we first consider the simplified case, where only one investor (say investor $n$) acts strategically, while the rest stay price-takers. Although there is a characterization of the unique equilibrium in the competitive setting, there is no general explicit formula for the equilibrium returns. This comes in sharp contrast to the frictionless case and makes the determination of the impact of investor $n$ to the equilibrium return (and hence her best-response) quite challenging. However, as we shall show below, the following assumption allows us to go beyond a characterization of the market's equilibirum returns and deal both with the best-response and the Nash equilibrium under transaction costs.
\begin{assumption}\label{ass:equal_tolerance}
Investors have the same risk tolerances: $\de_m=\bde$, for each $m\in\II$.
\end{assumption}

Moreover, for the sake of tractability we consider the following class:
\begin{flalign*}
&\Hci:=\Big\{X: \text{$X:[0,T]\times\Omega\rightarrow \mathbb{R}^{d}$ optional s.t. } \|X\|_{\Hci}:=\esssup_{\omega\in\Omega}\big\{\esssup_{t\in[0,T]}\|X(t)\|_{2}\big\}<\infty\Big\}. 
\end{flalign*}
Now, as it can be seen by \cite[Theorem A.4]{C}, \eqref{eq:fr_solution_compe} admits a (pathwise) unique solution. More precisely, the BSDE part of the aforementioned solution, assuming $\z_{m},\nu\in\Hci$ for all $1\leq m \leq N$, belongs to the following class of continuous semimartingales which forms a subspace of $\Hci$: 
\begin{equation*}
\Ac:=\Big\{X:dX(t)=a^{X}(t)dt+dM^{X}(t), \ \text{s.t. $X(T)=0$, $a^{X}\in\Hci$, $M^{X}$ is an $\Fbb$-MG in $\Hci$}\Big\}.
\end{equation*}
In turn, we define the set:
\begin{equation*}
\Wc_{\Ac}:=\Big\{X:\, dX(t)=\dot{X}(t)dt, \ t\in[0,T], \ X(0)=0 \text{ and }\dot{X}\in\Ac\Big\},
\end{equation*}
which forms the new admissible class of investor $n$'s demand in a market with both transaction costs and price impact. 

We hence update Assumption \ref{ass:standing} to:
\begin{assumption}\label{ass:standing2}
$\z_{m},\nu\in\Hci$ for all $m\in\II$ and $\phi_{n},\psi\in\Wc_{\Ac}$.
\end{assumption}

When considering the frictional analogue of the price-impact process for investor $n$  introduced in \eqref{eq:impact}, $\Ac$ becomes a natural choice for her control. In other words, by observing how the rest of the investors act optimally as price takers at the equilibrium, we have the following result:


\begin{lemma}\label{lem:fr_objective_impact}
Impose Assumptions \ref{ass:equal_tolerance} and \ref{ass:standing2}. In a market with transaction costs, when only investor $n$ acts strategically the objective functional \eqref{eq:fr_objective} becomes $\widetilde{\Fc}_{\La,n}:\Ac\rightarrow\Rbb$:
\begin{flalign}\label{eq:fr_objective_impact}
\widetilde{\Fc}_{\La,n}(\dot{\phi})&:=\Ebb\bigg[\int_0^{T}e^{-rt}\bigg((\phi(t))^{'}\Big(\frac{\Si(\z_{-n}(t)-\phi(t)-\psi(t))}{\bde(N-1)}+\frac{2\La}{N-1}(a^{\dot{\psi}}(t)-r\dot{\psi}(t))\Big) \notag \\
&-\frac{1}{2\bde}(\phi(t)+\zeta_{n}(t))^{'}\Si (\phi(t)+\zeta_{n}(t))-(\dot{\phi}(t))^{'}\Big(\La+\frac{2\La}{N-1}\Big)\dot{\phi}(t)\bigg)dt\bigg].
\end{flalign}
\end{lemma}
\begin{remark}
Comparison of \eqref{eq:fr_objective_impact} and \eqref{eq:best_response_objective} directly indicates the impact of the transaction costs on the investor's objective functional. While the term $\Si(\z_{-n}-\phi_{n}-\psi)/\bde(N-1)$ is common in both cases (assuming equal risk tolerances), transaction costs add two additional terms: 
one in the linear part which depends on noise demand $2\La(a^{\dot{\psi}}-r\dot{\psi})/(N-1)$ and another in the quadratic part (on top of $\dot{\phi}'\Lambda\dot{\phi}$). In particular, for the first term we get that if noise traders sell at a constant rate, positive investor's demand increases her utility due to price impact. This effects comes solely due to the presence of transaction costs. On the contrary, the second term implies that price impact does increase the effect of the transaction costs, even when there is no noise traders.  
\end{remark}
Based on a calculus of variations argument, we get the following result for the frictional best-response strategy of investor $n$:
\begin{proposition}\label{prop:mes}
Impose Assumptions \ref{ass:equal_tolerance} and \ref{ass:standing2}. The unique solution $(\pfp,\pfdp)$ of \eqref{eq:fr_objective_impact} is characterized by the following FBSDE:
\begin{equation}\label{eq:mes}
\begin{aligned}
&d\pfp(t)=\pfdp(t)dt, \ \ \ \pfp(0)=0, \\
&d\pfdp(t)=d\widetilde{M}_{n}(t)+B(\pfp(t)-\mathbb{TP}_{n}(t))dt+r\pfdp(t)dt, \ \ \ \pfdp(T)=0,
\end{aligned}
\end{equation}
where: 
\begin{equation}\label{eq:TP_and_B}
\mathbb{TP}_{n}:=\pnfp_{n}+\frac{2\bde\Si^{-1}\La}{N+1}(a^{\dot{\psi}_{n}}-r\dot{\psi}),\qquad B:=\frac{\La^{-1}\Si}{2\bde},
\end{equation}
$\pnfp_{n}$ is the best-response strategy in frictionless case, given in \eqref{eq:best_response_solution} and $\widetilde{M}_{n}\in\Hci$ is a continuous $\Fbb$-martingale derived as part of the solution.
\end{proposition}
The existence and uniqueness of the solution of \eqref{eq:mes} is stated in the following proposition. In fact, this FBSDE is a part of a larger class of linear FBSDEs that do admit  unique (global) solutions.
\begin{proposition}\label{prop:fr_best_response_solution}
Impose Assumptions \ref{ass:equal_tolerance} and \ref{ass:standing2}. The unique optimal demand $\pfp\in\Wc_{\Ac}$ that satisfies \eqref{eq:mes} admits the following explicit form:
\begin{equation}\label{eq:best_response_FBSDE}
\pfp(t)=\int_0^{t}e^{-\int_{s}^{t}F(u)du}\widetilde{\mathbb{TP}}_{n}(s)ds,
\end{equation}
where
\begin{align*}
&F(t):=-\Big(\Delta G(t)-\frac{r}{2}\dot{G}(t)\Big)^{-1}B\dot{G}(t),\quad \Delta :=B+\frac{r^2}{4}I_{d},\quad G(t):=\cosh(\sqrt{\Delta}(T-t))\\
&\widetilde{\mathbb{TP}}_{n}(t):=\Big(\Delta G(t)-\frac{r}{2}\dot{G}(t)\Big)^{-1}\Ebb\bigg[\int_{t}^{T}\Big(\Delta G(s)-\frac{r}{2}\dot{G}(s)\Big)Be^{-\frac{r}{2}(s-t)}\mathbb{TP}_{n}(s)ds\bigg|\Fc(t)\bigg]
\end{align*}
and $B$ is given in \eqref{eq:TP_and_B}.
\end{proposition}
Based on the above characterization, we are able to provide a closed-form expression for the equilibrium returns under transaction costs when only investor $n$ acts strategically and the rest of the investors are price takers.
\begin{corollary}\label{cor:fr_best_response_return}
Impose Assumptions \ref{ass:equal_tolerance} and \ref{ass:standing2}. In a market with transaction costs, when only investor $n$ acts strategically, the unique equilibrium returns in this context are given by:
\begin{equation}\label{eq:fr_best_response_return}
\Hci\ni\m_{\La,n}:=\mnfp+\frac{2N\La}{N^2-1}(a^{\dot{\psi}}-r\dot{\psi}),  
\end{equation}
where $\mnfp$ is the corresponding frictionless equilibrium given in \eqref{eq:best_response_return}.
\end{corollary}
\begin{remark}\label{rem:nonoise}
A direct consequence of \eqref{eq:fr_best_response_return} is that equal risk tolerances result in equilibrium returns under single investor's price impact with transaction costs being equal to the corresponding frictionless equilibrium if there is no noise traders (i.e.~no liquidity premium is created). This is in line with the price-taking equilibrium. Indeed, as shown in \cite{C}, quadratic transaction costs do not change the equilibrium returns as long as there is homogeneous risk aversion and absence of noise traders. This also holds when a single and (as we shall verify below) all investors act strategically. When noise traders are present, positive liquidity premium is created when the rate of noise demand is sufficiently high (or simply positive when $r=0$). In fact, comparing \eqref{eq:fr_best_response_return} to the Corollary 5.4 of \cite{C}, the strategic behaviour of a single investor increases the equilibrium return (and hence the liquidity premium) with transaction costs. From \eqref{eq:fr_best_response_return}, we also verify that $\m_{\La,n}$ converges to $\mnfp$ (i.e.~liquidity premium vanishes) when $\La$ goes to zero. 
\end{remark}

\subsection{The Nash equilibrium returns in a market with transaction costs}\label{sec:fr_Nash}

We are now ready to generalize the concepts presented in \S\ref{sub:2} in a market with transaction costs. To this end, we extend Assumption \ref{ass:standing2} to each $\phi_{m}, \  m\in \II$ and begin with a characterization of the investors' demands when all of them act strategically. More precisely, by shifting the investors' exposure, in accordance to Definition \ref{def:nash_frless}, the linear system \eqref{eq:Nash_system_frless} becomes the following system of linear FBSDEs, under the same boundary value conditions as \eqref{eq:mes}, for $ m\in \II$:
\begin{equation}\label{eq:Nash_system_fr}
\begin{aligned}
&d\pfnm(t)=\pfdnm(t)dt, \\
&d\pfdnm(t)=d\widetilde{M}_{m}(t)+B(\pfnm(t)-\breve{\mathbb{TP}}_{m}(t))dt+r\pfdnm(t)dt,
\end{aligned}
\end{equation}
where:
\begin{equation*}
\breve{\mathbb{TP}}_{m}:=\frac{\overbrace{\sum_{i=1\, ,i\neq m}^{N}\Big(\bde\Si^{-1}\nu-\breve{\phi}_{\La,i}\Big)}^{\z_{-m}(\breve{\phi}_{\La})}-\psi-(N-1)\z_{m}}{N+1}+\frac{2\bde\Si^{-1}\La}{N+1}(a^{\dot{\psi}_{n}}-r\dot{\psi}).
\end{equation*}
The uniqueness of the Nash equilibrium returns also holds in the market with transaction costs and is a consequence of the following lemma:
\begin{lemma}\label{lem:ks}
The system of linear FBSDEs \eqref{eq:Nash_system_fr} admits a unique solution such that $\pfnm \in\Wc_{\Ac}$, for all $m\in \II$.
\end{lemma}
With the above result in hand we are ready to derive the point at which the frictional market equilibrates when all the agents act strategically.
\begin{theorem}\label{thm:fr_Nash}
Impose Assumptions \ref{ass:equal_tolerance} and \ref{ass:standing2}. In a market with transaction costs, the unique Nash equilibrium returns are given by:
\begin{equation}\label{eq:fr_Nash}
\Hci\ni\mfn:=\mnfn+\frac{2\La}{N(N-1)}(a^{\dot{\psi}}-r\dot{\psi}),
\end{equation}
where $\mnfn$ is the corresponding frictionless Nash equilibrium (for equal risk tolerances) and $B$ is defined in \eqref{eq:TP_and_B}.
\end{theorem}
\begin{remark}\label{rem:nash}
Based on Remark \ref{rem:nonoise}, we expect that the absence of noise traders and equal risk tolerances would result in Nash equilibrium returns with transaction costs being equal to the frictionless competitive ones. Indeed, this stems directly from \eqref{eq:fr_Nash}. 
In fact, under common risk tolerance, we have seen from \eqref{eq:Nash_competitive_frless} that the absence of noise traders also implies that the frictionless Nash equilibrium equals to the competitive one too. Furthermore, comparing \eqref{eq:fr_Nash} to the frictional competitive equilibrium returns derived in \cite{C}, liquidity premium that stems from transaction costs is higher in the non-competitive case if and only if $N>2$. In words, when more than two strategic investors with the same risk tolerance participate in the market, the effect of transaction costs on liquidity premium is higher due their strategic behavior. In the special case of only two strategic investors the effect of transaction costs on liquidity premium is equal between competitive and non-competitive market's structure. Finally, as expected, when $\La$ tends to zero, liquidity premium due to transaction costs vanishes. 
\end{remark}

\subsection{The case of two investors}

Interestingly, Assumption \ref{ass:equal_tolerance} in Theorem \ref{thm:fr_Nash} can be dropped when considering a market of only two investors, i.e. $N=2$. In fact even in the case of different risk tolerances a strong connection between the frictional and frictionless Nash equilibrium returns remains. This is supported by the following result:
\begin{theorem}\label{thm:fr_Nash_two_inv}
Impose Assumption \ref{ass:standing2} and let $N=2$. In a market with transaction costs, the unique Nash equilibrium returns are given by:
\begin{equation}\label{eq:fr_Nash_two_inv}
\Hci\ni\mfn:=\mnfn+\La(a^{\dot{\psi}}-r\dot{\psi}),
\end{equation}
where $\mnfn$ is the corresponding frictionless Nash equilibrium.
\end{theorem}
Note that \eqref{eq:fr_Nash_two_inv} coincides with \eqref{eq:fr_Nash} when $\delta_1=\delta_2$. 

\newpage
\appendix
\section{Proofs of Section \ref{sec:Price_Impact_NoFri}}\label{sec:appendixA}

\begin{proof}[Proof of Proposition \ref{prop:best_response}]
The concavity of the functional $\widetilde{\Fc}_{n}(\phi)$ follows by its quadratic form. More formally, as shown in \cite[Theorem~6.2.1]{convex1}, a sufficient condition for strict concavity in this context is for the second order G\^{a}teaux differential to be strictly negative for every admissible direction. 
Using \eqref{eq:best_response_problem} we conclude that:
\begin{equation*}
\Big(D^2_{G}\widetilde{\Fc}_{n}(\phi)\theta_{n},\theta_{n}\Big)=\Ebb\Bigg[\int_0^{T}e^{-r t}\Big(-k_{n}(\theta_{n}(t))^{'}\Si\theta_{n}(t)\Big)dt\Bigg] < 0, \ \ \ \forall \ \theta_{n}\in\Hc.
\end{equation*}
By the above we have that $\widetilde{\Fc}_{n}(\phi)$ is strictly concave and therefore the solution of \eqref{eq:best_response_problem} is characterized by the first order condition through the G\^{a}teaux differential, i.e. (s.f. \cite[Proposition~II.2.1]{convex2}):
\begin{equation*}
\Big(D_{G}\widetilde{\Fc}_{n}(\phi),\theta_{n}\Big)=0, \ \text{for all} \ \theta_{n}\in\Hc.
\end{equation*}
For this, we calculate:
\begin{equation*}
\begin{aligned}
&\Big(D_{G}\widetilde{\Fc}_{n}(\phi),\theta_{n}\Big)=\Ebb\Bigg[\int_0^{T}e^{-r t}\bigg(-\frac{\psi(t)^{'}\Si}{\delta_{-n}}-k_{n}(\phi_{n}(t))^{'}\Si-\frac{(\z_{n}(t))^{'}\Si}{\de_{n}}+\frac{(\z_{-n}(t))^{'}\Si}{\de_{-n}}\bigg)\theta_{n}(t)dt\Bigg],
\end{aligned}
\end{equation*}
which holds for all admissible directions. Therefore, we get:
\begin{equation*}
-\frac{\psi(t)^{'}\Si}{\de_{-n}}-k_{n}(\phi_{n}(t))^{'}\Si-\frac{(\z_{n}(t))^{'}\Si}{\de_{n}}+\frac{(\z_{-n}(t))^{'}\Si}{\de_{-n}}=0,
\end{equation*}
which in turn leads to the following representation of the frictionless optimizer under the price impact of investor $n$: 
\begin{equation}\label{eq:3.5}
\frac{\z_{-n}(t)-\psi(t)-\frac{\de_{-n}}{\de_{n}}\z_{n}(t)}{\de_{-n}k_{n}}.
\end{equation}
Equivalently $\delta_{-n}k_{n}=(\delta_{n}+\delta)/\delta_{n}$. Hence, \eqref{eq:3.5} rearranges to the desired result, which belongs in $\Hc$ by linearity.
\end{proof}

\begin{proof}[Proof of Proposition \ref{prop:best_response_revealed}]
Recalling \eqref{eq:best_response_solution}, \eqref{eq:best_response_return}, \eqref{eq:revealed_exposure} and evaluating the revealed exposure process of investor $n$ for $(\pnfp_{n},\m_{n})$, we get:
\begin{flalign*}
\Si\z_{n}(t;\pnfp)&=\de_{n}\overbrace{\bigg(\frac{\de_{n}\m_{-n}(t)+\de\m(t)}{\de_{n}+\de}\bigg)}^{\mnfp}-\Si\frac{\z_{-n}(t)-\psi(t)-\frac{\de_{-n}}{\de_{n}}\z_{n}(t)}{\de_{-n}k_{n}} \\
&=\frac{\de_{n}\m_{-n}(t)}{\de_{-n}k_{n}}+\frac{\Si\z_{n}(t)}{\de_{-n}k_{n}}+\frac{\frac{\delta_{-n}}{\de_{n}}\Si\z_{n}(t)}{\de_{-n}k_{n}}.
\end{flalign*}
Therefore, we have:
\begin{equation}\label{eq:3.13}
\z_{n}(t;\pnfp)=\frac{\frac{\z_{-n}(t)-\psi(t)}{(1/\de_{n})\de_{-n}}+\z_{n}(t)+\frac{\de_{-n}}{\de_{n}}\z_{n}(t)}{(\de+\de_{n})/\de_{n}}.
\end{equation}
Equivalently, \eqref{eq:3.13} can be written as:
\begin{flalign*}
\z_{n}(t;\pnfp)&=\frac{\frac{\z(t)-\psi(t)}{(1/\de_{n})\delta_{-n}}-\frac{\z_{n}(t)}{(1/\de_{n})\delta_{-n}}+\z_{n}(t)+\frac{\de_{-n}}{\de_{n}}\z_{n}(t)}{\delta/\de_{n}+1} \\
&=\frac{\la_{n}^2}{1-\la_{n}^2}(\z_{-n}(t)-\psi(t))+\frac{1}{1+\la_{n}}\z_{n}(t),
\end{flalign*}
which concludes the proof.
\end{proof}

\begin{proof}[Proof of Corollary \ref{cor:fr_strategies_equality}]
Recalling \eqref{eq:best_response_solution} and \eqref{eq:revealed_exposure}, we get that relation  \eqref{eq:best_response_revealed} can be equivalently written as:
\begin{equation*}
\pnfp_{n}(t)=\de_{n}\Si^{-1}\mnfp(t)-\bigg(\frac{\la_{n}^2}{1-\la_{n}^2}(\z_{-n}(t)-\psi(t))+\frac{1}{1+\la_{n}}\z_{n}(t)\bigg).
\end{equation*}
Furthermore, evaluating \eqref{eq:frless_comp_strategies} at $\m$, yields:
\begin{equation*}
\pnf_{n}(t)=\de_{n}\Si^{-1}\m(t)-\z_{n}(t).
\end{equation*}
Thus, for the ``if" implication using \eqref{eq:best_response_return} and \eqref{eq:3.13} we have:
\begin{flalign*}
0&=\de_{n}\Si^{-1}\Bigg(\frac{\de_{n}\Big(\frac{\Si(\z_{-n}(t)-\psi(t))}{\de_{-n}}-\frac{\Si(\z(t)-\psi(t))}{\de}\Big)}{\de_{n}+\de}\Bigg)-\z_{n}(t;\pnfp)+\z_{n}(t) \\
&=-\de_{n}\Si^{-1}\m(t)-\frac{\de_{-n}}{\de_{n}}\z_{n}(t)+\frac{\de_{-n}+\de_{n}}{\de_{n}}\z_{n}(t) \\
&=\de_{n}\Si^{-1}\m(t)-\z_{n}(t).
\end{flalign*}
For the other direction, we have that:
\begin{flalign*}
0&=\de_{n}\Si^{-1}\frac{\Si(\z(t)-\psi(t))}{\de}-\z_{n}(t) \\
\z_{-n}(t)-\psi(t)&=\frac{\de_{-n}}{\de_{n}}\z_{n}(t),
\end{flalign*}
which together with \eqref{eq:3.5} finishes the proof.
\end{proof}

\begin{proof}[Proof of Theorem \ref{the:Nash_frless}]
Let $\Hc\ni\pnfn_{m}, \ m\in\II$ be the unique solution set of \eqref{eq:Nash_system_frless}. Focusing on investor $n$ and evaluating \eqref{eq:revealed_exposure} at \eqref{eq:Nash_system_frless}, we see that her revealed exposure becomes:
\begin{flalign*}
\z_{n}(t;\pnfn)&=\de_{n}\Si^{-1}\nu(t)-\underbrace{\frac{\la_{n}(\z_{-n}(t;\pnfn)-\psi(t))-\la_{-n}\z_{n}(t)}{\la_{n}+1}}_{(\star)}.
\end{flalign*}
We now use the investors' symmetry to derive the aggregate revealed exposure process $\z(t;\pnfn):=\sum_{m=1}^{N}\z_{m}(t;\pnfn)$. Therefore, adding and substracting $\z_{n}(t;\pnfn)/\de_{-n}k_{n}$ and using the equivalent formulation of \eqref{eq:3.5} for the term $(\star)$, we have:
\begin{flalign*}
\underbrace{\frac{\overbrace{(\delta_{-n}k_{n}-1)}^{\de/\de_{n}}}{\delta_{-n}k_{n}}}_{>0}\z_{n}(t;\pnfn)&=\de_{n}\Si^{-1}\nu(t)-\frac{\z(t;\pnfn)}{\de_{-n}k_{n}}+\frac{\psi(t)}{\de_{-n}k_{n}}+\frac{(\de_{-n}/\de_{n})\z_{n}(t)}{\de_{-n}k_{n}} \\
\sum_{m=1}^{N}\z_{m}(t;\pnfn)&=\Si^{-1}\nu(t)\sum_{m=1}^{N}\frac{1}{(1/\de_{m})^2\de}+\Si^{-1}\nu(t)\de-\z(t;\pnfn)\sum_{m=1}^{N}\frac{1}{(1/\de_{m})\delta}\\
&+\psi(t)\sum_{m=1}^{N}\frac{1}{(1/\de_{m})\delta}+\frac{1}{\delta}\sum_{m=1}^{N}\delta_{-m}\zeta_{m}(t).
\end{flalign*}
Therefore, we have for the aggregate exposure process:
\begin{equation}\label{eq:3.16}
\z(t;\pnfn)=\frac{\de\Si^{-1}\nu(t)\Big(1+\sum_{m=1}^{N}\la_{m}^2\Big)+\sum_{m=1}^{N}\la_{-m}\z_{m}(t)+\psi(t)}{2}
\end{equation}
Hence, when all the investors act strategically the market clearing condition becomes:
\begin{equation}\label{eq:3.17}
\pnfn_1(t)+\cdots+\pnfn_{N}(t)+\psi(t)=0, \ \ \ \forall \ t\in[0,T].
\end{equation}
Now by using the fact that $\pnfn_{m}$ for each $m$ (which solves \eqref{eq:Nash_system_frless}) can be equivalently written as: $\delta_{m}\Si^{-1}\nu-\z_{m}(\pnfn)$ for each $m$, and using \eqref{eq:3.16}, we get that \eqref{eq:3.17} becomes:
\begin{flalign*}
0&=\sum_{m=1}^{N}\Big(\de_{m}\Si^{-1}\nu(t)-\z_{m}(t;\pnfn)\Big)+\psi(t) \\
\de\Si^{-1}\nu(t)-\de\Si^{-1}\nu(t)\sum_{m=1}^{N}\la_{m}^2&=\sum_{m=1}^{N}(1-\la_{m})\z_{m}(t)-\psi(t) \\
\nu(t)-\nu(t)\sum_{m=1}^{N}\la_{m}^2&=\frac{\Si}{\de}\Big(\z(t)-\psi(t)-\sum_{m=1}^{N}\la_{m}\z_{m}(t)\Big).
\end{flalign*}
Noting that the symmetric matrix $\Big(1-\sum_{m=1}^{N}\la_{m}^2\Big)I_{d\times d}$ is positive definite, we arrive at the desired result. 

The uniqueness of the Nash equilibrium returns is a direct consequence of the uniqueness of the solution set $(\pnfn_1,...,\pnfn_{N})$ of the linear system \eqref{eq:Nash_system_frless}. This stems from the fact that the system's coefficient matrix is invertible, which can be easily determined by checking that its transpose is strictly diagonally dominant. Lastly, by linearity we get that the frictionless Nash equilibrium indeed belongs in $\Hc$.
\end{proof}

\begin{proof}[Proof of Proposition \ref{prop:surplus_comparison}]
We begin with the calculation of the asymptotic surplus at the frictionless Nash equilibrium, i.e.: $\mathbb{US}_{1}^{\mnfn}(\pnfn)$ as $\de_1\rightarrow\infty$. To this end, we first need to derive investor's Nash optimal demand, $\pnfn_{1}$. This can be readily obtained by noting that \eqref{eq:Nash_system_frless} can be decoupled through the use of the equilibrium condition: $\sum_{m=1}^{N}\pnfn_{m}+\psi=0$. Hence, we have: 
\begin{flalign}
\pnfn_1&=\frac{\la_1(\z_{-1}(\pnfn)-\psi)-\la_{-1}\z_1}{\la_1+1}=\frac{\la_1(\de_{-1}\Si^{-1}\nu+\pnfn_1)-\la_{-1}\z_1}{\la_1+1} \notag \\
&=\la_1\de_{-1}\Si^{-1}\nu-\la_{-1}\z_1 \label{eq:appendix_surplus_1}.
\end{flalign}
Furthermore, note that for the frictionless Nash equilibrium returns, $\mnfn$, we have:
\begin{equation}\label{eq:appendix_surplus_2}
\lim\limits_{\de_1\rightarrow\infty}\mnfn=\frac{\Si(\z_{-1}-\psi)}{2\de_{-1}}.
\end{equation}
Using \eqref{eq:appendix_surplus_1} and \eqref{eq:appendix_surplus_2}, we readily get the following limit for investor's Nash optimal demand at $\mnfn$:
\begin{equation*}
\lim\limits_{\de_1\rightarrow\infty}\pnfn_1=\frac{1}{2}(\z_{-1}-\psi).
\end{equation*}
We now apply the Dominated Convergence Theorem (DCT) on $dt\otimes d\Pbf$ to justify bringing the limits (w.r.t.~$\de_1$) inside the product integral in $\mathbb{US}_{1}^{\mnfn}(\pnfn)$ (passing between the sequential characterization of limits and back when necessary). To this end, note that $\psi,\z_{m}\in\Hc$ for each $m$ and:
\begin{flalign*}
\|\mnfn\|_2&=\bigg\|\frac{\Si}{\de}\frac{(\z-\psi)-\sum\limits_{m=1}^{N}\la_{m}\z_{m}}{1-\sum\limits_{m=1}^{N}\la_{m}^2}\bigg\|_2 
\leq \frac{1}{\de}\frac{1}{1-\la_1^2-\sum\limits_{m=2}^{N}\la_{m}^2}\underbrace{(\|\Si\|_2\|\z-\psi\|_2+\sum\limits_{m=1}^{N}\|\z_{m}\|_2)}_{=:C} \\
&< \frac{1}{\de\sum\limits_{m=2}^{N}\la_{m}(1-\la_{m})}C <\frac{1}{\de_2\big(1-\frac{\de_2}{\de}\big)}C <\frac{1}{\de_2\big(1-\frac{\de_2}{\de_{-1}}\big)}C.
\end{flalign*} 
With regards to $\pnfn$, we use \eqref{eq:appendix_surplus_1} and the above to conclude that $\|\pnfn_1\|_2\leq\de_{-1}\|\Si^{-1}\|_2\|\mnfn\|_2+\|\z_1\|_2$. Hence,  applying DCT gives the desired result for the asymptotic surplus at the frictionless Nash equilibrium. 

Following similar arguments, we deal with the competitive case. More precisely:
\[\|\m\|_2\leq \frac{1}{\de_{-1}}\|\Si\|_2\|\z-\psi \|_2 \qquad\text{ and }\qquad
\|\pnf_1\|_2\leq \|\z-\psi\|_2+\|\z_1\|_2. \] 
Taking into account that $\lim\limits_{\de_1\rightarrow\infty}\pnf_1'\m=0$ and applying the DCT once more conclude the proof.
\end{proof}

\section{Proofs of Section \ref{sec:Frictions}}\label{sec:appendixB}

\begin{proof}[Proof of Lemma \ref{lem:fr_objective_impact}]
We begin by formally generalizing \eqref{eq:impact}. For this, we work similarly as in the frictionless market and introduce the following conditions for all $(\phi_{n},\dot{\phi}_{n})$:
\begin{align}
&\phi_{n}(t)+\sum_{\substack{m=1\\m\neq n}}^{N}\plm(t)+\psi(t)=0, \ \ \ \forall \ t\in[0,T], \ \forall \ \phi_{n}\in\Wc_{\Ac}, \label{eq:4.2} \\
&\dot{\phi}_{n}(t)+\sum_{\substack{m=1\\m\neq n}}^{N}\pdlm(t)+\dot{\psi}(t)=0,\ \ \ \forall \ t\in[0,T], \label{eq:4.3} 
\end{align}
where the pair $(\plm,\pdlm)$, for $m\in\II\setminus\{n\}$ solves \eqref{eq:fr_solution_compe}. Thus, \eqref{eq:4.3} reads as:
\begin{flalign*}
0&=a^{\dot{\phi}_{n}}(t)dt+\frac{\La^{-1}}{2}\sum_{\substack{m=1\\m\neq n}}^{N}\bigg(\frac{\Si\plm(t)}{\bde}-\Big(\nu(t)-\frac{\Si\z_{m}(t)}{\bde}\Big)\bigg)dt+r\sum_{\substack{m=1\\m\neq n}}^{N}\pdlm(t)dt \\
&+a^{\dot{\psi}}(t)dt+dM^{\dot{\phi}_{n}}(t)+dM^{\dot{\psi}}(t)+\underbrace{\sum_{\substack{m=1\\m\neq n}}^{M}dM_{m}(t)}_{=:dM(t)}.
\end{flalign*}
By using \eqref{eq:4.3} again, we have that $r\sum_{\substack{m=1\\m\neq n}}^{N}\pdlm(t)=-r(\dot{\phi}_{n}(t)+\dot{\psi}(t))$ which in turn yields:
\begin{flalign*}
0&=a^{\dot{\phi}_{n}}(t)dt+\frac{\La^{-1}\Si}{2}\sum_{\substack{m=1\\m\neq n}}^{N}\frac{\plm(t)}{\bde}dt-\frac{\La^{-1}(N-1)}{2}\nu(t)dt+\frac{\La^{-1}\Si}{2}\sum_{\substack{m=1\\m\neq n}}^{N}\frac{\z_{m}(t)}{\bde}dt\\
&-r(\dot{\phi}_{n}(t)+\dot{\psi}(t))dt+a^{\dot{\psi}}(t)dt+\underbrace{dM^{\dot{\phi}_{n}}(t)+dM(t)+dM^{\dot{\psi}}(t)}_{=:dM^{S}(t)}.
\end{flalign*}
Thanks to equal risk tolerances and \eqref{eq:4.2} we get that $\sum_{\substack{m=1\\m\neq n}}^{N}\plm(t)/\bde=-(\phi_{n}(t)+\psi(t))/\bde$, which gives:
\begin{flalign*}
0&=\Big(\frac{2\La}{N-1}a^{\dot{\phi}_{n}}(t)-\frac{2r\La}{N-1}\dot{\phi}_{n}(t)+\frac{\Si(\z_{-n}(t)-\phi_{n}(t)-\psi(t))}{\bde(N-1)}\notag\\
&+\frac{2\La}{N-1}(a^{\dot{\psi}}(t)-r\dot{\psi}(t))-\nu(t)\Big)dt+\frac{2\La}{N-1}dM^{S}(t).
\end{flalign*}
Since the rest of the terms are absolutely continuous, $dM^{S}(t)$ vanishes since $M^{S}$ is a continuous martingale of finite variation (from the stronger absolute continuity). Therefore, we have that the price-impact process of investor $n$, $\nu(\phi_{n},\dot{\phi}_{n})\in\Hci$ in a market with transaction costs is:
\begin{equation}\label{ks1}
\nu(t;\phi_{n},\dot{\phi}_{n})=\Big(\frac{2\La}{N-1}a^{\dot{\phi}_{n}}(t)-\frac{2r\La}{N-1}\dot{\phi}_{n}(t)+\frac{\Si(\z_{-n}(t)-\phi_{n}(t)-\psi(t))}{\bde(N-1)}+\frac{2\La}{N-1}(a^{\dot{\psi}}(t)-r\dot{\psi}(t))\Big).
\end{equation}
The next step is to incorporate \eqref{ks1} into \eqref{eq:fr_objective} in order to introduce price impact in the frictional problem (similarly to what we did for the frictionless case). To this end, note that we have:
\begin{flalign}\label{ks1.5}
\Ebb\bigg[\int_0^{T}e^{-rt}(\phi_{n}(t))^{'}\nu(t) dt\bigg]&=\Ebb\bigg[\int_0^{T}e^{-rt}(\phi_{n}(t))^{'}\Big(\frac{2\La}{N-1}a^{\dot{\phi}_{n}}(t)-\frac{2r\La}{N-1}\dot{\phi}_{n}(t) \notag \\
&+\frac{\Si(\z_{-n}(t)-\phi_{n}(t)-\psi(t))}{\bde(N-1)}+\frac{2\La}{N-1}(a^{\dot{\psi}}(t)-r\dot{\psi}(t))\Big)dt\bigg].
\end{flalign}
We now want assert the following equality:
\begin{equation}\label{ks2}
\begin{aligned}
&\Ebb\bigg[\int_0^{T}e^{-rt}(\phi_{n}(t))^{'}\frac{2\La}{N-1}d\dot{\phi}_{n}(t)\bigg]=\Ebb\bigg[\int_0^{T}e^{-rt}(\phi_{n}(t))^{'}\frac{2\La}{N-1}a^{\dot{\phi}_{n}}(t)dt\bigg].
\end{aligned}
\end{equation}
We begin by examining the well-posedness of the above semimartingale integral. For the finite variation part, we require Lebesgue-Stieltjes integrability, i.e. for the pathwise absolutely continuous $dA^{\dot{\phi}_{n}}_{i}(t):=a^{\dot{\phi}_{n}}_{i}(t)dt$, $i\in\DD$ and $t\in[0,T]$:
\begin{flalign*}
\Ebb\bigg[\int_0^{t}e^{-rt}\frac{2\La^{i,i}}{N-1}|\phi_{n}^{i}(s)|d|A^{\dot{\phi}_{n}}_{i}|(s)\bigg]&=\Ebb\bigg[\int_0^{t}e^{-rt}\frac{2\La^{i,i}}{N-1}|\phi_{n}^{i}(s)a^{\dot{\phi}_{n}}_{i}(s)|ds\bigg], \\
&\leq\frac{2\La^{i,i}}{N-1}\esssup_{\omega\in\Omega}\esssup_{s\in[0,T]}|\phi_{n}^{i}(s)|\Ebb\bigg[\int_0^{T}|a^{\dot{\phi}_{n}}_{i}(s)|ds\bigg]<\infty,
\end{flalign*}
where the above hold by absolute continuity and $\phi_{n}\in\Wc_{\Ac}$, $a^{\dot{\phi}_{n}}\in\Hci$. For the martingale part, we note that $M^{\dot{\phi}_{n}}$ has continuous paths and belongs in $\Hci$. Furthermore, $e^{-rt}(\phi_{n}(t))^{'}2\La/(N-1)$ is It\^{o}-integrable with respect to $M^{\dot{\phi}_{n}}$ (componentwise), since it also belongs in $\Hci$. Therefore, by martingale preservation $\int_0^{t}e^{-rs}(\phi_{n}(s))^{'}2\La/(N-1)dM^{\dot{\phi}_{n}}(s)$ is a continuous martingale, zero at $t=0$. More precisely, we also have $\int_0^{t}e^{-rs}(\phi_{n}(s))^{'}2\La/(N-1)dM^{\dot{\phi}_{n}}(s)\in\Hci$ since for each $i\in\DD$ integration by parts yields:
\begin{flalign}\label{ksref} 
\bigg|\int_0^{t}e^{-rs}\frac{2\La^{i,i}}{N-1}\phi_{n}^{i}(s)dM^{\dot{\phi}_{n}}_{i}(s)\bigg|&=\frac{2\La^{i,i}}{N-1}\bigg|e^{-rt}\phi_{n}^{i}(t)M^{\dot{\phi}_{n}}_{i}(t)-\int_0^{t}e^{-rs}M^{\dot{\phi}_{n}}_{i}(s)(\dot{\phi}_{n}^{i}(s)-r\phi_{n}^{i}(s))ds\bigg|.
\end{flalign}
So, the left-hand side of \eqref{ks2} becomes:

\begin{flalign*}
&\Ebb\bigg[\int_0^{T}e^{-rt}(\phi_{n}(t))^{'}\frac{2\La}{N-1}d\dot{\phi}_{n}(t)\bigg]=\Ebb\bigg[\sum_{i=1}^{d}\int_0^{T}e^{-rt}\phi_{n}^{i}(t)\frac{2\La^{i,i}}{N-1}d\dot{\phi}_{n}^{i}(t)\bigg]\\
&=\Ebb\bigg[\sum_{i=1}^{d}\int_0^{T}e^{-rt}\phi_{n}^{i}(t)\frac{2\La^{i,i}}{N-1}dA^{\dot{\phi}_{n}}_{i}(t)+ \sum_{i=1}^{d}\int_0^{T}e^{-rt}\phi_{n}^{i}(t)\frac{2\La^{i,i}}{N-1}dM_{i}^{\dot{\phi}_{n}}(t)\bigg] \\
&=\Ebb\bigg[\sum_{i=1}^{d}\int_0^{T}e^{-rt}\phi_{n}^{i}(t)\frac{2\La^{i,i}}{N-1}a^{\dot{\phi}_{n}}_{i}(t)dt\bigg]=\Ebb\bigg[\int_0^{T}e^{-rt}(\phi_{n}(t))^{'}\frac{2\La}{N-1}a^{\dot{\phi}_{n}}(t)dt\bigg],
\end{flalign*}
Having established \eqref{ks2}, we apply integration by parts on the left-hand side of \eqref{ks2}, yielding:
\begin{flalign}\label{ks2.5}
&\Ebb\bigg[\int_0^{T}e^{-rt}(\phi_{n}(t))^{'}\frac{2r\La}{N-1}\dot{\phi}_{n}(t)-e^{-rt}(\dot{\phi}_{n}(t))^{'}\frac{2\La}{N-1}\dot{\phi}_{n}(t)dt\bigg]=\Ebb\bigg[\int_0^{T}e^{-rt}(\phi_{n}(t))^{'}\frac{2\La}{N-1}a^{\dot{\phi}_{n}}(t)dt\bigg].
\end{flalign}
Thus, taking into account \eqref{ks1.5}, \eqref{ks2.5} and adjusting \eqref{eq:fr_objective}, we get the desired result.
\end{proof}

\begin{proof}[Proof of Proposition \ref{prop:mes}]
The concavity of \eqref{eq:fr_objective_impact} is readily ensured by its quadratic form. More formally using a sufficient condition for strict concavity, as we did in Proposition \ref{prop:best_response}, through the negativeness of the functional's second G\^{a}teaux differential, we conclude that for all admissible directions $\dot{\theta}_{n}$ s.t. $\theta_{n}\in\Wc_{\Ac}$, we have:
\begin{flalign*}
\Big(D^2_{G}\widetilde{\Fc}_{\La,n}(\dot{\phi})\dot{\theta}_{n},\dot{\theta}_{n}\Big)&=\Ebb\Bigg[\int_0^{'}e^{-rt}\bigg(-(\theta_{n}(t))^{'}\frac{2\Si}{\bde(N-1)}\theta_{n}(t)-(\dot{\theta}_{n}(t))^{'}\frac{4\La}{N-1}\dot{\theta}_{n}(t)\\
&-(\theta_{n}(t))^{'}\frac{\Si}{\bde}\theta_{n}(t)-(\dot{\theta}_{n}(t))^{'}2\La\dot{\theta}_{n}(t)\bigg)dt\Bigg]< 0,
\end{flalign*}
for $\bar{\de}, 1/\bar{\de}(N-1)>0$ and positive definite matrices $\Si$ and $\La$. By the above, we have that $\widetilde{\Fc}_{\La,n}(\dot{\phi})$ is strictly concave and therefore the solution of \eqref{eq:fr_objective_impact} is characterized by the first order condition through the G\^{a}teaux differential, i.e. for all $\dot{\theta}_{n}$ as prescribed above we have that:
\begin{equation*}
\Big(D_{G}\widetilde{\Fc}_{\La,n}(\dot{\phi}),\dot{\theta}_{n}\Big)=0.
\end{equation*}
Applying this to \eqref{eq:fr_objective_impact}, we get:
\begin{flalign*}
\Big(D_{G}\widetilde{\Fc}_{\La,n}(\dot{\phi})&,\dot{\theta}_{n}\Big)=\Ebb\bigg[\int_0^{T}e^{-rt}\bigg(\Big(\frac{(\z_{-n}(t)-2\phi_{n}(t)-\psi(t))^{'}\Si}{\bde(N-1)}+(a^{\dot{\psi}}(t)-r\dot{\psi}(t))^{'}\frac{2\La}{N-1}\\
&-(\phi_{n}(t)+\z_{n}(t))^{'}\frac{\Si}{\bde}\Big)\theta_{n}(t)-(\dot{\phi}_{n}(t))^{'}\Big(\frac{4\La}{N-1}+2\La\Big)\dot{\theta}_{n}(t)\bigg)dt\bigg] \\
&=\Ebb\bigg[\int_0^{T}e^{-rt}\bigg(\Big(\pnfp_{n}(t)+\frac{2\bde\Si^{-1}\La}{N+1}(a^{\dot{\psi}_{n}}(t)-r\dot{\psi}(t))-\phi_{n}(t)\Big)^{'}\frac{(N+1)\Si}{\bde(N-1)}\int_0^{t}\dot{\theta}_{n}(s)ds\\
&-(\dot{\phi}_{n}(t))^{'}\frac{2(N+1)\La}{N-1}\dot{\theta}_{n}(t)\bigg)dt\bigg] \\
&=\Ebb\bigg[\int_0^{T}\bigg(\Ebb\bigg[\int_{0}^{T}f(s)ds\bigg|\Fc(t)\bigg]-\int_{0}^{t}f(s)ds-e^{-rt}(\dot{\phi}_{n}(t))^{'}\frac{2(N+1)\La}{N-1}\bigg)\dot{\theta}_{n}(t)dt\bigg]=0,
\end{flalign*}  
where $\pnfp_{n}$ is given in \eqref{eq:best_response_solution} and 
$$f(t):=e^{-rt}\big(\pnfp_{n}(t)+\frac{2\bde\Si^{-1}\La}{N+1}(a^{\dot{\psi}_{n}}(t)-r\dot{\psi}(t))-\phi_{n}(t)\big)^{'}\frac{(N+1)\Si}{\bde(N-1)}.$$ 
The third equality above holds by Fubini-Tonelli and the tower properties of conditional expectation. Now, since the above has to hold for every admissible direction $\dot{\theta}_{n}$, $D_{G}\widetilde{\Fc}_{\La,n}(\dot{\phi})\in\Ac$, we have that \eqref{eq:fr_objective_impact} has a unique solution if and only if the following stochastic integral equation is uniquely satisfied:
\begin{equation}\label{ks3}
\dot{\phi}_{n}(t)=Be^{rt}\Ebb\bigg[\int_{t}^{T}e^{-rs}\Big(\pnfp_{n}(s)+\frac{2\bde\Si^{-1}\La}{N+1}(a^{\dot{\psi}_{n}}(s)-r\dot{\psi}(s))-\phi_{n}(s)\Big)ds\bigg|\Fc(t)\bigg],
\end{equation}
where $B:=\La^{-1}\Si/2\bde$. In turn, assuming \eqref{ks3} admits a (unique) solution and taking the cadlag version of the $\Fbb$-martingale 
$$\Ebb\bigg[\int_{0}^{T}e^{-rs}\Big(\pnfp_{n}(s)+\frac{2\bde\Si^{-1}\La}{N+1}(a^{\dot{\psi}_{n}}(s)-r\dot{\psi}(s))-\phi_{n}(s)\Big)ds\bigg|\Fc(t)\bigg],$$ 
denoted as $\hat{M}_{n}$ \footnote{Note that $\hat{M}_{n}\in\Hci$ and in particular it has continuous paths, by martingale representation.}, through integration by parts we can rewrite \eqref{ks3} as:
\begin{equation*}
d\dot{\phi}_{n}(t)=e^{rt}d\hat{M}_{n}(t)+B\Big(\phi_{n}(t)-\pnfp(t)-\frac{2\bde\Si^{-1}\La}{N+1}(a^{\dot{\psi}_{n}}(t)-r\dot{\psi}(t))\Big)dt+r\dot{\phi}_{n}(t)dt.
\end{equation*}
By the preservation of the martingale property $d\widetilde{M}_{n}(t):=e^{rt}d\hat{M}_{n}(t)$ is a $\Fbb$-martingale with continuous paths and belongs in $\Hci$ (shown through integration by parts as in \eqref{ksref}). 
Hence, combining the above with $d\phi_{n}(t)=\dot{\phi}_{n}(t)dt$ leads to the desired FBSDE representation.

Conversely assume that \eqref{eq:mes} admits a unique solution, then integration by parts yields for $t\in[0,T]$:
\begin{equation}\label{ks4}
e^{-rt}\dot{\phi}_{n}(t)=\dot{\phi}_{n}(0)+\int_0^{t}e^{-rs}d\widetilde{M}_{n}(s)+\int_0^{t}e^{-rs}B\Big(\phi_{n}(s)-\pnfp(s)-\frac{2\bde\Si^{-1}\La}{N+1}(a^{\dot{\psi}_{n}}(s)-r\dot{\psi}(s))\Big)ds.
\end{equation}
Now, by the terminal condition $\dot{\phi}_{n}(T)=0$, we get:
\begin{equation}\label{ks5}
\dot{\phi}_{n}(0)=-\int_0^{T}e^{-rs}d\widetilde{M}_{n}(s)-\int_0^{T}e^{-rs}B\Big(\phi_{n}(s)-\pnfp(s)-\frac{2\bde\Si^{-1}\La}{N+1}(a^{\dot{\psi}_{n}}(s)-r\dot{\psi}(s))\Big)ds.
\end{equation}
Substituting \eqref{ks5} into \eqref{ks4}, rearranging and taking conditional expectations we get \eqref{ks3}.
\end{proof}

\begin{proof}[Proof of Proposition \ref{prop:fr_best_response_solution}]
The FBSDE of \eqref{eq:mes} is a part of larger class of linear FBSDEs that have the following form:
\begin{equation}\label{ks6}
\begin{aligned}
&dX(t)=\dot{X}(t)dt, \ \ \ X(0)=0, \\
&d\dot{X}(t)=dN(t)+B(X(t)-\xi(t))dt+r\dot{X}(t)dt, \ \ \ \dot{X}(T)=0,
\end{aligned}
\end{equation}
where $B$ is a real matrix with only positive eigenvalues, $r\geq 0$, $\xi\in\Hc$ and $N$ is a square integrable $\Fbb$-(continuous) martingale. The existence and uniqueness of this class of FBSDEs' solutions is studied in \cite[Theorem~A.4]{C}, where it is shown that the (pathwise) unique processes $(X,\dot{X})$ that solve \eqref{ks6} satisfy a linear differential equation with the following (unique) solution given through variation of parameters and (primary) matrix functions:
\begin{flalign}
&X(t)=\int_0^{t}e^{-\int_{s}^{t}F(u)du}\widetilde{\xi}(s)ds, \label{ks7}\\
&\dot{X}(t)=\widetilde{\xi}(t)-F(t)X(t), \label{ks8}
\end{flalign}
where:
\begin{flalign*}
&\Delta:=B+\frac{r^2}{4}I_{d}, \\
&G(t):=\cosh(\sqrt{\Delta}(T-t)), \\
&F(t):=-\Big(\Delta G(t)-\frac{r}{2}\dot{G}(t)\Big)^{-1}B\dot{G}(t), \\
&\widetilde{\xi}(t)=\Big(\Delta G(t)-\frac{r}{2}\dot{G}(t)\Big)^{-1}\Ebb\bigg[\int_{t}^{T}\Big(\Delta G(s)-\frac{r}{2}\dot{G}(s)\Big)Be^{-\frac{r}{2}(s-t)}\xi(s)ds\bigg|\Fc(t)\bigg],
\end{flalign*}
More precisely, we choose $\widetilde{\xi}$ to be the continuous version of $(\Delta G(t)-\frac{r}{2}\dot{G}(t))^{-1}\Ebb[\int_{t}^{T}(\Delta G(s)-\frac{r}{2}\dot{G}(s))Be^{-\frac{r}{2}(s-t)}\xi(s)ds|\Fc(t)]$, i.e.:
\begin{flalign*}
\Big(\Delta G(t)-\frac{r}{2}\dot{G}(t)\Big)^{-1}\Ebb\bigg[\int_{t}^{T}\Big(\Delta G(s)-\frac{r}{2}\dot{G}(s)\Big)Be^{-\frac{r}{2}(s-t)}\xi(s)ds\bigg|\Fc(t)\bigg]&=\\
\Big(\Delta G(t)-\frac{r}{2}\dot{G}(t)\Big)^{-1}e^{\frac{r}{2}t}\bigg(\bar{M}(t)-\int_{0}^{t}\Big(\Delta G(s)-\frac{r}{2}\dot{G}(s)\Big)Be^{-\frac{r}{2}s}\xi(s)ds\bigg),
\end{flalign*}
where $\bar{M}$ is a cadlag version of $\Ebb[\int_{0}^{T}(\Delta G(s)-\frac{r}{2}\dot{G}(s))Be^{-\frac{r}{2}s}\xi(s)ds|\Fc(t)]$ (ensured by the right continuity of the filtration), which in particular has continuous paths by martingale representation. 

Going back to \eqref{eq:mes} and applying the above we get that a sufficient condition for $\pfp\in\Hci$ is $\mathbb{TP}_{n}\in\Hci$, which holds since $\mathbb{TP}_{n}$ is a linear combination of processes in $\Hci$. To see that, initially note that $\widetilde{\mathbb{TP}}_{n}\in\Hci$ since $\mathbb{TP}_{n}\in\Hci$ and:
\begin{equation*}\sup_{t\in[0,T]}\|(\Delta G(t)-r/2\dot{G}(t))^{-1}\|_{2}<\infty,\end{equation*}
by \cite[Lemma B.2]{C} and the fact that:
\begin{equation*}\inf_{t\in[0,T]}x\cosh(x(T-t))+(r/2)x\sinh(x(T-t))\geq x>0,\quad\forall x\in(0,\infty).\end{equation*}
 In turn, these yield $\pfp\in\Hci$ as desired. Furthermore, we also get $\pfdp\in\Hci$ by \eqref{ks8} and the triangle inequality. Lastly, to see that the unique solution of the RDE: $Y(t):=\int_0^{t}e^{-\int_{s}^{t}F(u)du}\widetilde{\mathbb{TP}}_{n}(s)ds$ in \eqref{ks8} indeed satisfies the FBSDE of \eqref{eq:mes}, initially note that the terminal condition $\dot{Y}(T)=0$ holds since $\dot{G}(T)=0$. Writing \eqref{ks8} explicitly for $Y$, we have (naturally the arbitrary $\xi$ is also replaced with $\mathbb{TP}_{n}$ in $\bar{M}$):
\begin{flalign*}
\dot{Y}(t)&=\Big(\Delta G(t)-\frac{r}{2}\dot{G}(t)\Big)^{-1}e^{\frac{r}{2}t}\bigg(\bar{M}(t)-\int_{0}^{t}\Big(\Delta G(s)-\frac{r}{2}\dot{G}(s)\Big)Be^{-\frac{r}{2}s}\mathbb{TP}_{n}(s)ds\bigg)-F(t)Y(t).
\end{flalign*}
Multiplying with $(\Delta G(t)-\frac{r}{2}\dot{G}(t))e^{-\frac{r}{2}t}$ and using the fact that $B=\De-r^2/4$ we get:
\begin{flalign*}
\Big(\Delta G(t)-\frac{r}{2}\dot{G}(t)\Big)e^{-\frac{r}{2}t}\dot{Y}(t)&=\bar{M}(t)-\int_0^{t}\Big(\Delta G(s)-\frac{r}{2}\dot{G}(s)\Big)Be^{-\frac{r}{2}s}\mathbb{TP}_{n}(s)ds\\
&+\Big(\De\dot{G}(t)-\frac{r^2}{4}\dot{G}(t)\Big)e^{-\frac{r}{2}t}Y(t).
\end{flalign*}
Setting $\bar{Y}(t):=e^{-\frac{r}{2}t}Y(t)$, with $\dot{\bar{Y}}(t)=(-r/2)\bar{Y}(t)+e^{-\frac{r}{2}t}\dot{Y}(t)$ and $\bar{\mathbb{TP}}_{n}(t):=e^{-\frac{r}{2}t}\mathbb{TP}_{n}(t)$ we get:
\begin{flalign*}
\Big(\Delta G(t)-\frac{r}{2}\dot{G}(t)\Big)\dot{\bar{Y}}(t)&=\bar{M}(t)-\int_0^{t}\Big(\Delta G(s)-\frac{r}{2}\dot{G}(s)\Big)B\bar{\mathbb{TP}}_{n}(s)ds+\Big(\De\dot{G}(t)-\frac{r}{2}\De G(t)\Big)\bar{Y}(t).
\end{flalign*}
In differential form the above reads:
\begin{flalign*}
&\Big(\De \dot{G}(t)-\frac{r}{2}\ddot{G}(t)\Big)\dot{\bar{Y}}(t)dt+\Big(\De G(t)-\frac{r}{2}\dot{G}(t)\Big)d\dot{\bar{Y}}(t)=d\bar{M}(t)-\Big(\De G(t)-\frac{r}{2}\dot{G}(t)\Big)B\bar{\mathbb{TP}}_{n}(t)dt \\
&+\Big(\De\ddot{G}(t)-\frac{r}{2}\De\dot{G}(t)\Big)\bar{Y}(t)dt+\Big(\De\dot{G}(t)-\frac{r}{2}\De G(t)\Big)\dot{\bar{Y}}(t)dt.
\end{flalign*}
Noting that $\ddot{G}(t)=\De G(t)$, as well as that $\De$ commutes with $G(t)$ and $ \dot{G}(t)$, ``multiplying" with $(\De G(t)-r/2\dot{G}(t))^{-1}$ and using associativity we get:
\begin{flalign*}
d\dot{\bar{Y}}(t)=\Big(\De G(t)-\frac{r}{2}\dot{G}(t)\Big)^{-1}d\bar{M}(t)-B\bar{\mathbb{TP}}_{n}(t)dt+\De\bar{Y}(t)dt.
\end{flalign*}
Now, $\bar{M}$ is a continuous martingale that belongs in $\Hci$ (since $\bar{\mathbb{TP}}_{n}$ does) and $\sup_{t\in[0,T]}\|(\Delta G(t)-r/2\dot{G}(t))^{-1}\|_{2}<\infty$. Hence, by martingale preservation $dM(t):=(\Delta G(t)-r/2\dot{G}(t))^{-1}d\bar{M}(t)$ is a continuous martingale that belongs in $\Hci$ (the latter can be shown similarly as \eqref{ksref}, through integration by parts). Furthermore, using the definition of $\De$ we get:
\begin{flalign*}
d\dot{\bar{Y}}(t)=dM(t)+B(\bar{Y}(t)-\bar{\mathbb{TP}}_{n}(t))dt+\frac{r^2}{4}\bar{Y}(t)dt.
\end{flalign*}
Using the fact that $d\dot{\bar{Y}}(t)=(-r/2)\dot{\bar{Y}}(t)dt+e^{-\frac{r}{2}t}d\dot{Y}(t)-r/2e^{-\frac{r}{2}t}\dot{Y}(t)dt$ and ``multiplying" by $e^{\frac{r}{2}t}$ we get:
\begin{flalign*}
d\dot{Y}(t)=e^{\frac{r}{2}t}dM(t)+B(Y(t)-\mathbb{TP}_{n}(t))dt+r\dot{Y}(t)dt.
\end{flalign*}
Defining $d\widetilde{M}(t):=e^{\frac{r}{2}t}dM(t)$ and arguing as above, we have that $\widetilde{M}$ is a continuous martingale that belongs in $\Hci$. Hence $Y$ does indeed satisfy the dynamics of \eqref{eq:mes}, with $\widetilde{M}$ determined as part of the solution. This completes the proof.
\end{proof}

\begin{proof}[Proof of Corollary \ref{cor:fr_best_response_return}]
Under transaction costs, the market-clearing condition holds for both the investors' demands and trading rates. Now, following similar steps as the ones we did to derive the price-impact process \eqref{ks1}, we may use \eqref{eq:mes} for the strategic investor $n$ and get:
\begin{equation*}
\nu(t)=\frac{N^2}{N^2-1}\m(t)-\frac{\Si}{\bde(N^2-1)}\z_{n}(t)+\frac{2N\La}{N^2-1}(a^{\dot{\psi}}(t)-r\dot{\psi}(t)),
\end{equation*}
where $\m$ is the frictionless competitive equilibrium returns. Uniqueness follows from the uniqueness of the frictional best-response strategy for investor $n$, as well as from the uniqueness of the competitive frictional strategies for the rest of the investors. Now, recalling \eqref{eq:best_response_return} and noting that the above is a linear combination of processes in $\Hci$, we get the desired result.
\end{proof}
\begin{proof}[Proof of Lemma \ref{lem:ks}]
Note that for 
$$\phi=\begin{bmatrix}\phi_1^{'} & \cdots & \phi_{N}^{'}\end{bmatrix}^{'},\quad \dot{\phi}=\begin{bmatrix}\dot{\phi}_1^{'} & \cdots & \dot{\phi}_{N}^{'}\end{bmatrix}^{'}\quad \text{and}\quad \widetilde{M}=\begin{bmatrix}\widetilde{M}_1^{'} & \cdots & \widetilde{M}_{N}^{'}\end{bmatrix}^{'}$$
the system \eqref{eq:Nash_system_fr} can be rewritten as:
\begin{equation}\label{ks10}
\begin{aligned}
&d\phi(t)=\dot{\phi}(t)dt, \\
&d\dot{\phi}(t)=d\widetilde{M}(t)+\mathbf{B}(\phi(t)-\mathbf{B}^{-1}Z(t))+r\dot{\phi}(t),
\end{aligned}
\end{equation}
where:
\begin{flalign*}
&\mathbf{B}:=\begin{bmatrix}
B       && \frac{B}{N+1} && \ldots  && \frac{B}{N+1} \\
\frac{B}{N+1} &&  B      &&\ddots  &&\vdots\\
\vdots  &&\ddots  &&B       && \frac{B}{N+1} \\
\frac{B}{N+1} &&\ldots  && \frac{B}{N+1} && B
\end{bmatrix}\in\mathbb{R}^{dN\times dN} \text{ and } \\
&Z:=\begin{bmatrix}\Big(B\frac{(N-1)\bde\Si^{-1}\nu-\psi-(N-1)\z_i+2\bde\Si^{-1}\La(a^{\dot{\psi}}-r\dot{\psi})}{N+1}\Big)^{'} \end{bmatrix}^{'}_{i\in\II}.
\end{flalign*}
Now note that for:
\begin{flalign*}
&C:=\begin{bmatrix}
1       && \frac{1}{N+1} && \ldots  && \frac{1}{N+1} \\
\frac{1}{N+1} &&  1      &&\ddots  &&\vdots\\
\vdots  &&\ddots  &&1       && \frac{1}{N+1} \\
\frac{1}{N+1} &&\ldots  && \frac{1}{N+1} && 1
\end{bmatrix}\in\mathbb{R}^{N\times N}, \\
&\mathbf{B}=C\otimes B,
\end{flalign*}
where $\otimes$ denotes the Kronecker product. By the above, $\mathbf{B}^{-1}$ is well defined as $\det(\mathbf{B})=\det(C\otimes B)=\det(C)^{d}\det(B)^{N}\neq 0$, since $C$ is real, symmetric, strictly diagonally dominant with positive diagonal entries (and hence symmetric positive definite) and $B$ is the product of two symmetric positive definite matrices. Furthermore, the eigenvalues, $(b_{i})_{i\in\II}$, of $B$ are all positive (as the product of two symmetric positive definite matrices) and the eigenvalues, $(c_{j})_{i\in\II}$ of $C$ are again positive (as a symmetric positive definite matrix). Therefore, the eigenvalues of $\mathbf{B}$ given as: $b_{i}c_{j}, \ i,j\in\II$ are also positive. We also have that $\widetilde{M}\in\Hci$ and $\mathbf{B}^{-1}Z\in\Hci$ as $\nu,\z_{m}\in\Hci,$ for each $m\in\II$ and $\psi\in\Wc_{\Ac}$. So, \eqref{ks10} is of the form of \eqref{ks6} (assuming no dependence of $\nu$ on $\phi_{m},\dot{\phi}_{m}, \ m\in\II$) and thus it admits a unique solution satisfying the conditions of $\Wc_{\Ac}$ (which can be directly shown, following the proof of Proposition \ref{prop:fr_best_response_solution}). 
\end{proof}
\begin{proof}[Proof of Theorem \ref{thm:fr_Nash}]
Let $(\pfnm,\pfdnm), \ m\in\II$ be the optimal demand and trading rate of \eqref{eq:Nash_system_fr} whose existence and uniqueness is shown by Lemma \ref{lem:ks}. In order for the market clearing to hold, the following equilibrium conditions must be satisfied at all times: 
\begin{align}
&\sum_{m=1}^{N}\pfnm(t)+\psi(t)=0, \label{eq:4.17} \\
&\sum_{m=1}^{N}\pfdnm(t)+\dot{\psi}(t)=0. \label{eq:4.18}
\end{align}
More analytically, \eqref{eq:4.18} develops as:
\begin{flalign*}
0&=\sum_{m=1}^{N}d\widetilde{M}_{m}(t)+B\sum_{m=1}^{N}\pfnm(t)dt\\
&-B\sum_{m=1}^{N}\frac{\sum_{\substack{i=1\\i\neq m}}^{N}(\bde\Si^{-1}\nu(t)-\breve{\phi}_{\La,i}(t))-\psi(t)-(N-1)\z_{m}(t)+B^{-1}(a^{\dot{\psi}}(t)-r\dot{\psi}(t))}{N+1}dt \\
&+r\sum_{m=1}^{N}\pfdnm(t)dt+a^{\dot{\psi}}(t)dt+dM^{\dot{\psi}}(t).
\end{flalign*}
Using \eqref{eq:4.17} we get:
\begin{flalign*}
0&=\overbrace{\sum_{m=1}^{N}d\widetilde{M}_{m}(t)+dM^{\dot{\psi}}(t)}^{=:d\widetilde{M}^{S}(t)}-B\psi(t)dt-r\dot{\psi}(t)dt+a^{\dot{\psi}}(t)dt \\
&-B\sum_{m=1}^{N}\frac{(N-1)\bde\Si^{-1}\nu(t)+\pfnm(t)-(N-1)\z_{m}(t)+B^{-1}(a^{\dot{\psi}}(t)-r\dot{\psi}(t))}{N+1}dt.
\end{flalign*}
Using \eqref{eq:4.17} once more, we get:
\begin{flalign*}
0&=d\widetilde{M}^{S}(t)+\frac{(N-1)B\z(t)}{N+1}dt-\frac{NB\psi(t)}{N+1}dt+\frac{a^{\dot{\psi}}(t)-r\dot{\psi}(t)}{N+1}dt-\frac{BN(N-1)\bde\Si^{-1}\nu(t)}{N+1}dt
\end{flalign*}
Arguing as in the proof of Lemma \ref{lem:fr_objective_impact} for \eqref{ks1}, $d\widetilde{M}^{S}(t)$ vanishes. The above is rewritten as:
\begin{flalign*}
&\frac{BN(N-1)\bde\Si^{-1}\nu(t)}{N+1}=\frac{(N-1)B(\z(t)-\psi(t))}{N+1}-\frac{B\psi(t)}{N+1}+\frac{a^{\dot{\psi}}(t)-r\dot{\psi}(t)}{N+1}
\end{flalign*}
and by rearranging the terms we get the desired result. 

The uniqueness of the frictional Nash equilibrium returns is a direct consequence to the uniqueness of the solution set $(\pfnm,\pfdnm), \ m\in\II$. Lastly, the frictional Nash equilibrium is in $\Hci$ by linearity and the fact that $\z_{m}\in\Hci$, for each $m\in\II$ and $\psi\in\Wc_{\Ac}$.
\end{proof}
\begin{proof}[Proof of Theorem \ref{thm:fr_Nash_two_inv}]
The steps for this proof are similar to the ones of Lemma \ref{lem:fr_objective_impact}, Proposition \ref{prop:mes}, Proposition \ref{prop:fr_best_response_solution}, Lemma \ref{lem:ks} and Theorem \ref{thm:fr_Nash}. Thus, we omit a part of the analysis and refer to the original proofs for more details. 

We first consider the price-impact process of investor $1$. As in Lemma \ref{lem:fr_objective_impact}, the respective form of \eqref{ks1} in this case is given as:
\begin{equation*}
\nu(t;\phi_1,\dot{\phi}_1)=2\La a^{\dot{\phi}_1}(t)-2r\La\dot{\phi}_1(t)+\frac{\Si}{\de_2}(\z_2(t)-\phi_1(t)-\psi(t))+2\La(a^{\dot{\psi}}(t)-r\dot{\psi}(t)).
\end{equation*}
Using similar arguments as with \eqref{ks2} and \eqref{ks2.5} in Lemma \ref{lem:fr_objective_impact}, the adjusted frictional objective functional under investor $1$'s price impact becomes:
\begin{flalign*}
\widetilde{\Fc}_{\La,1}(\dot{\phi})&:=\Ebb\bigg[\int_0^{T}e^{-rt}\bigg((\phi_{1}(t))^{'}\Big(\frac{\Si(\z_{2}(t)-\phi_{1}(t)-\psi(t))}{\de_2}+2\La(a^{\dot{\psi}}(t)-r\dot{\psi}(t))\Big)  \\
&-\frac{1}{2\de_1}(\phi_{1}(t)+\zeta_{1}(t))^{'}\Si (\phi_{1}(t)+\zeta_{1}(t))-(\dot{\phi}_{1}(t))^{'}\Big(\La+2\La\Big)\dot{\phi}_{1}(t)\bigg)dt\bigg].
\end{flalign*}
Optimizing $\widetilde{\Fc}_{\La,1}(\dot{\phi})$, as in Proposition \ref{prop:mes}, we get that the unique solution of the above functional is characterized by the following FBSDE:
\begin{equation*}
\begin{aligned}
&d\widetilde{\phi}_{\La,1}(t)=\dot{\widetilde{\phi}}_{\La,1}(t)dt, \ \ \ \widetilde{\phi}_{\La,1}(0)=0, \\
&d\dot{\widetilde{\phi}}_{\La,1}(t)=d\widetilde{M}_{1}(t)+B_1(\widetilde{\phi}_{\La,1}(t)-\chi_{1}(t))dt+r\dot{\widetilde{\phi}}_{\La,1}(t)dt, \ \ \ \dot{\widetilde{\phi}}_{\La,1}(T)=0,
\end{aligned}
\end{equation*}
where: 
\begin{equation*}
\Hci\ni\chi_{1}:=\pnfp_{1}+\frac{2\de_1\de_2\Si^{-1}\La}{\de+\de_1}(a^{\dot{\psi}_{n}}-r\dot{\psi}),\qquad B_1:=\frac{\de+\de_1}{3\de_1\de_2}\frac{\La^{-1}\Si}{2},
\end{equation*}
$\pnfp_{1}$ is the best-response strategy of investor $1$ in frictionless case and $\widetilde{M}_{1}$ is a continuous MG that belongs in $\Hci$. In fact, noting that $B_1$ has only positive eigenvalues and following similar arguments as in Proposition \ref{prop:fr_best_response_solution}, the above FBSDE admits a unique (pathwise) solution. Using the investors' symmetry and recalling the characterization of the Nash equilibrium returns in \S\ref{sec:fr_Nash}, we get the following system for the frictional Nash optimal demands of the investors $ m =1, 2$:
\begin{equation}\label{eq:Nash_opt_demand_two}
\begin{aligned}
&d\pfnm(t)=\pfdnm(t)dt, \ \ \ \pfnm(0)=0, \\
&d\pfdnm(t)=d\widetilde{M}_{m}(t)+B_{m}(\pfnm(t)-\breve{\chi}_{m}(t))dt+r\pfdnm(t)dt, \ \ \ \pfdnm(T)=0,
\end{aligned}
\end{equation}
where $B_{m}:=(\de+\de_{m})\La^{-1}\Si/6\de_{m}\de_{-m}$ and:
\begin{equation*}
\breve{\chi}_{m}:=\frac{\la_{m}\Big(\overbrace{\sum\limits_{i=1,i\neq m}^{2}(\de_{i}\Si^{-1}\nu-\breve{\phi}_{\La,i})}^{\z_{-m}(\breve{\phi}_{\La})}-\psi\Big)-\la_{-m}\z_{m}}{\la_{m}+1}+\frac{2\de_{m}\de_{-m}\Si^{-1}\La}{\de+\de_{m}}(a^{\dot{\psi}_{n}}-r\dot{\psi}).
\end{equation*}
The key ingredient to derive the frictional Nash equilibrium returns in this setting is to derive the solution of \eqref{eq:Nash_opt_demand_two}. Note that for $\phi=\begin{bmatrix}\phi_1^{'} & \phi_{2}^{'}\end{bmatrix}^{'}$, $\dot{\phi}=\begin{bmatrix}\dot{\phi}_1^{'} &  \dot{\phi}_{2}^{'}\end{bmatrix}^{'}$ and $\widetilde{M}=\begin{bmatrix}\widetilde{M}_1^{'} & \widetilde{M}_{2}^{'}\end{bmatrix}^{'}$ the system of \eqref{eq:Nash_opt_demand_two} can be rewritten as:
\begin{equation*}
\begin{aligned}
&d\phi(t)=\dot{\phi}(t)dt, \\
&d\dot{\phi}(t)=d\widetilde{M}(t)+\mathbf{B}(\phi(t)-\mathbf{B}^{-1}Z(t))+r\dot{\phi}(t),
\end{aligned}
\end{equation*}
where:
\begin{flalign*}
&\mathbf{B}:=\begin{bmatrix}
\frac{\de+\de_1}{3\de_1\de_2}\frac{\La^{-1}\Si}{2} && \frac{\de_1}{3\de_1\de_2}\frac{\La^{-1}\Si}{2} \\
\frac{\de_2}{3\de_1\de_2}\frac{\La^{-1}\Si}{2} && \frac{\de+\de_2}{3\de_1\de_2}\frac{\La^{-1}\Si}{2}
\end{bmatrix}\in\mathbb{R}^{2d\times 2d}, \\
&Z:=\begin{bmatrix}\Big(B_1\frac{\de_1(\de_2\Si^{-1}\nu-\psi)-\de_2\z_1+2\de_1\de_2\Si^{-1}\La(a^{\dot{\psi}}-r\dot{\psi})}{\de+\de_1}\Big)^{'} &  \Big(B_2\frac{\de_2(\de_1\Si^{-1}\nu-\psi)-\de_1\z_2+2\de_1\de_2\Si^{-1}\La(a^{\dot{\psi}}-r\dot{\psi})}{\de+\de_2}\Big)^{'}\end{bmatrix}^{'}.
\end{flalign*}
Now note that for:
\begin{flalign*}
&C:=\begin{bmatrix}
\de+\de_1 && \de_1 \\
\de_2 && \de+\de_2
\end{bmatrix}\in\mathbb{R}^{2\times 2}, \\
&\mathbf{B}=C\otimes \frac{\La^{-1}\Si}{6\de_1\de_2}.
\end{flalign*}
But this is the setting studied in Proposition \ref{prop:mes} (assuming that there is not a dependence of $\nu$ on $\phi_{m},\dot{\phi}_{m}, \ m\in\II$). Hence, following similar steps we get that \eqref{eq:Nash_opt_demand_two} has a (pathwise) unique solution. 

Let $(\pfnm,\pfdnm), \  m \in\II$ be the unique optimal demand and trading rate of \eqref{eq:Nash_opt_demand_two}. In order for the market clearing to hold \eqref{eq:4.17} and \eqref{eq:4.18}, as shown in Theorem \ref{thm:fr_Nash} must be satisfied at all times. More precisely, using the above, \eqref{eq:4.18} reads: 
\begin{flalign*}
0&=\overbrace{\sum\limits_{i=1}^{2}d\widetilde{M}_{i}(t)+dM^{\dot{\psi}}(t)}^{=:dM^{S}}+B_1(\breve{\phi}_{\La,1}(t)-\breve{\chi}_1(t))dt+B_2(\breve{\phi}_{\La,2}(t)-\breve{\chi}_2(t))dt+r(\dot{\breve{\phi}}_{\La,1}(t)\\
&+\dot{\breve{\phi}}_{\La,2}(t))dt+a^{\dot{\psi}}(t)dt.
\end{flalign*}
Using \eqref{eq:4.17}, we have:
\begin{flalign*}
0&=dM^{S}(t)+\frac{\de\La^{-1}\Si}{3\de_1\de_2}(\breve{\phi}_{\La,1}(t)+\breve{\phi}_{\La,2}(t))dt+\frac{\La^{-1}\Si}{6\de_1\de_2}(\de_2\z_1(t)+\de_1\z_2(t))dt+\frac{\de\La^{-1}\Si}{6\de_1\de_2}\psi(t)dt\\
&+\frac{a^{\dot{\psi}}(t)-r\dot{\psi}(t)}{3}dt-\frac{\La^{-1}\nu(t)}{3}dt.
\end{flalign*}
Working as in the proof of Lemma \ref{lem:fr_objective_impact} for \eqref{ks1}, $d\widetilde{M}^{S}(t)$ vanishes. Hence, we get:
\begin{flalign*}
\nu(t)&=\frac{\Si}{2\de_1\de_2}(\de_2\z_1(t)+\de_1\z_2(t))-\frac{\de\Si}{2\de_1\de_2}\psi(t)+\La(a^{\dot{\psi}}(t)-r\dot{\psi}(t)) \\
&=\frac{\Si}{\de}\frac{\la_2\z_1(t)+\la_1\z_2(t)}{2\la_1\la_2}-\frac{\Si}{\de}\frac{\psi(t)}{2\la_1\la_2}+\La(a^{\dot{\psi}}(t)-r\dot{\psi}(t)).
\end{flalign*}
Noting that $\z-\la_1\z_1-\la_2\z_2=\la_2\z_1+\la_1\z_2$ and $1-\la_1^2-\la_2^2=2\la_1\la_2$ yields the desired result. The uniqueness of the frictional Nash equilibrium returns is a direct consequence to the uniqueness of the solution set $(\pfnm,\pfdnm), \  m\in\II$. Lastly, the frictional Nash equilibrium is in $\Hci$ by linearity and the fact that $\z_{1},\z_2\in\Hci$ and $\psi\in\Wc_{\Ac}$.
\end{proof}

\bigskip

\bibliographystyle{amsalpha}
\bibliography{paperbib}

\providecommand{\bysame}{\leavevmode\hbox to3em{\hrulefill}\thinspace}
\providecommand{\MR}{\relax\ifhmode\unskip\space\fi MR }
\providecommand{\MRhref}[2]{%
  \href{http://www.ams.org/mathscinet-getitem?mr=#1}{#2}
}
\providecommand{\href}[2]{#2}
\begin{thebibliography}{BFHMK18}

\bibitem[AC00]{AlmCh00}
R.~Almgren and N.~Chriss, \emph{Optimal execution of portfolio transactions},
  Journal of Risk \textbf{3} (2000), 5--39.

\bibitem[AK17]{B}
M.~Anthropelos and C.~Kardaras, \emph{Equilibrium in risk-sharing games},
  Finance and Stochastics (2017), no.~21 (3), 815–865.

\bibitem[AK24]{AnthKard24}
\bysame, \emph{Price impact under heterogeneous beliefs and restricted
  participation}, Journal of Economic Theory \textbf{215} (2024), 105774.

\bibitem[Ant17]{A}
M.~Anthropelos, \emph{The effect of market power on risk-sharing}, Mathematics
  and Financial Economics (2017), no.~11, 323--368.

\bibitem[ATHL05]{AlmThu05}
R.~Almgren, C.~Thum, E.~Hauptmann, and H.~Li, \emph{Equity market impact}, Risk
  (2005), 57--62.

\bibitem[Bas00]{Bas00}
S.~Basak, \emph{A model of dynamic equilibrium asset pricing with heterogeneous
  beliefs and extraneous risk}, Journal of Economic Dynamics and Control
  \textbf{24} (2000), 63--95.

\bibitem[BD19]{BusDum19}
A.~Buss and B.~Dumas, \emph{The dynamic properties of financial-market
  equilibrium with trading fees}, The Journal of Finance \textbf{74} (2019),
  no.~2, 795--844.

\bibitem[BFHMK18]{C}
B.~Bouchard, M.~Fukasawa, M.~Herdegen, and J.~Muhle-Karbe, \emph{Equilibrium
  returns with transaction costs}, Finance and Stochastics (2018), no.~22 (3),
  569--601.

\bibitem[CL24]{CorLil24}
F.~Cordoni and F.~Lillo, \emph{Instabilities in multi-asset and multi-agent
  market impact games}, Annals of Operations Research \textbf{336} (2024),
  505–539.

\bibitem[DGP05]{DuffGarPed05}
D.~Duffie, N.~Gârleanu, and L.H. Pedersen, \emph{Over-the-counter markets},
  Econometrica \textbf{73} (2005), no.~6, 1815--1847.

\bibitem[ET99]{convex2}
I.~Ekeland and R.~Temam, \emph{Convex analysis and variational problems}, SIAM,
  Philadelphia, 1999.

\bibitem[FIM18]{FraIsrMosk18}
A.~Frazzini, R.~Israel, and T.~Moskowitz, \emph{Transaction costs}, Working
  paper, available at SSRN: https://ssrn.com/abstract=3229719, 2018.

\bibitem[HHLT21]{HuHofLanTim21}
H.~Hau, P.~Hoffmann, S.~Langfield, and Y.~Timmer, \emph{Discriminatory pricing
  of over-the-counter derivatives}, Management Science \textbf{67} (2021),
  no.~11, 6660--6677.

\bibitem[HMKP21]{D}
M.~Herdegen, J.~Muhle-Karbe, and D.~Possama\"{i}, \emph{Equilibrium asset
  pricing with transaction costs}, Finance and Stochastics (2021), no.~25,
  1--45.

\bibitem[HW04]{HubSta04}
G.~Huberman and S.~Werner, \emph{Price manipulation and quasi-arbitrage},
  Econometrica \textbf{72} (2004), no.~4, 1247--1275.

\bibitem[KY19]{KoiYog19}
R.~S.~J. Koijen and M.~Yogo, \emph{A demand system approach to asset pricing},
  Journal of Political Economy \textbf{127} (2019), no.~4, 1475--1515.

\bibitem[Kyl85]{Kyl85}
A.S. Kyle, \emph{Continuous auctions and insider trading}, Econometrica
  \textbf{53} (1985), no.~6, 1315--1335.

\bibitem[KZ05]{convex1}
A.~Kurdila and M.~Zabarankin, \emph{Convex functional analysis}, Birkhauser
  Verlag, 2005.

\bibitem[LS19]{LuoSch19}
X.~Luo and A.~Schied, \emph{Nash equilibrium for risk-averse investors in a
  market impact game with transient price impact}, Market Microstructure and
  Liquidity \textbf{05} (2019), 2050001.

\bibitem[MR17]{MalRos17}
S.~Malamud and M.~Rostek, \emph{Decentralized exchange}, American Economic
  Review \textbf{107} (2017), no.~11, 3320--62.

\bibitem[RW15]{RosWer15}
M.~Rostek and M.~Weretka, \emph{{Dynamic thin markets}}, Review of Financial
  Studies \textbf{28} (2015), 2946--2992.

\bibitem[RY22]{RosYoo22}
M.~Rostek and J.~H. Yoon, \emph{Equilibrium theory of financial markets: Recent
  developments}, Prepared for The Journal of Economic Literature, 2022.

\bibitem[SZ19]{SchZha19}
A.~Schied and T.~Zhang, \emph{A market impact game under transient price
  impact}, Mathematics of Operations Research \textbf{44} (2019), no.~1,
  102--121.

\bibitem[Vay98]{Vay98}
D.~Vayanos, \emph{Transaction costs and asset prices: A dynamic equilibrium
  model}, The Review of Financial Studies \textbf{11} (1998), no.~1, 1--58.

\bibitem[Vay99]{Vay99}
\bysame, \emph{Strategic trading and welfare in a dynamic market}, Review of
  Economic Studies \textbf{66} (1999), no.~2, 219--54.

\bibitem[Viv11]{Viv11}
X.~Vives, \emph{Strategic supply function competition with private
  information}, Econometrica \textbf{79} (2011), 1919--1966.

\bibitem[VV99]{VayVil99}
D.~Vayanos and JL~Vila, \emph{Equilibrium interest rate and liquidity premium
  with transaction costs}, Economic Theory \textbf{13} (1999), 517–539.

\bibitem[Web23]{price_impact}
K.~Webster, \emph{Handbook of price impact modeling}, CRC Press, Boca Raton,
  FL, 2023.

\bibitem[Wes18]{West18}
K.~Weston, \emph{Existence of a radner equilibrium in a model with transaction
  costs}, Mathematics and Financial Economics \textbf{12} (2018), 517–539.

\end{thebibliography}
\end{document}